\documentclass[english,pra,notitlepage,twocolumn,nofootinbib]{revtex4-1}
\usepackage[T1]{fontenc}
\usepackage[utf8]{inputenc}
\usepackage{xcolor}
\usepackage{array}
\usepackage{verbatim}
\usepackage{units}
\usepackage{mathtools}
\usepackage{amsmath}
\usepackage{amsthm}
\usepackage{amssymb}
\usepackage{graphicx}
\usepackage{cancel}

\makeatletter

\providecommand{\tabularnewline}{\\}

\theoremstyle{plain}
\newtheorem{lem}{\protect\lemmaname}
\theoremstyle{plain}
\newtheorem{conjecture}{\protect\conjecturename}
\theoremstyle{plain}

\theoremstyle{definition}
 \newtheorem{example}{\protect\examplename}
 \theoremstyle{definition}
 
 \theoremstyle{plain}
\newtheorem{proposition}{\protect\propositionname}

\usepackage{braket}
\usepackage{bbm}

\allowdisplaybreaks

\newcommand{\tr}{\operatorname{tr}}
\newcommand{\id}{\mathbbm{1}}

\let\oldsqrt\sqrt
\def\sqrt{\mathpalette\DHLhksqrt}
\def\DHLhksqrt#1#2{%
\setbox0=\hbox{$#1\oldsqrt{#2\,}$}\dimen0=\ht0
\advance\dimen0-0.2\ht0
\setbox2=\hbox{\vrule height\ht0 depth -\dimen0}%
{\box0\lower0.4pt\box2}}

\let\originalleft\left
\let\originalright\right
\renewcommand{\left}{\mathopen{}\mathclose\bgroup\originalleft}
\renewcommand{\right}{\aftergroup\egroup\originalright}

\renewcommand\bra[1]{{\langle{#1}|}}
\makeatletter
\renewcommand\ket[1]{%
  \@ifnextchar\bra{\k@t{#1}\!}{\k@t{#1}}%
}
\newcommand\k@t[1]{{|{#1}\rangle}}
\makeatother

\usepackage{bm}
\usepackage[unicode=true, breaklinks=false, pdfborder={0 0 1}, backref=false, colorlinks=true, linkcolor=blue, urlcolor=blue, citecolor=blue,bookmarksopen=true]{hyperref}

\let\OLDthebibliography\thebibliography
\renewcommand\thebibliography[1]{
	\OLDthebibliography{#1}
	\setlength{\parskip}{0pt}
	\setlength{\itemsep}{0pt plus 0.3ex}
}

\makeatother

\usepackage{babel}
\providecommand{\conjecturename}{Conjecture}
\providecommand{\corollaryname}{Corollary}
\providecommand{\examplename}{Example}
\providecommand{\lemmaname}{Lemma}
\providecommand{\remarkname}{Remark}
\providecommand{\propositionname{Proposition}}

\begin{document}

\title{Basis-independent system-environment coherence is necessary to detect magnetic field direction in an avian-inspired quantum magnetic sensor}

\author{Thao P. Le}
\email{thao.le.16@ucl.ac.uk}

\affiliation{Dept. of Physics and Astronomy, University College London, Gower Street, London WC1E 6BT}

\author{Alexandra Olaya-Castro}
\email{a.olaya@ucl.ac.uk}
\affiliation{Dept. of Physics and Astronomy, University College London, Gower Street, London WC1E 6BT}

\date{\today}

\begin{abstract}
Advancing our understanding of non-trivial quantum effects in biomolecular complexes operating in physiological conditions requires the precise characterisation of the non-classicalities that may be present in such systems as well as asserting whether such features are required for robust function. Here we consider an avian-inspired quantum magnetic sensor composed of two radicals with a third ``scavenger'' radical under the influence of a collisional environment that allows to capture a variety of decoherence processes. We show that basis-independent coherence, in which the initial system-environment state is non-maximally mixed, is necessary for optimal performance of the quantum magnetic sensor, and appears to be sufficient in particular situations. We discuss how such non-maximally mixed initial states may be common for a variety of biomolecular scenarios.  Our results therefore suggest that a small degree of coherence\textemdash regardless of basis\textemdash is likely to be a quantum resource for biomolecular systems operating at the interface between the quantum and classical domains.
\end{abstract}
\maketitle

\section{Introduction}
The origin and relevance of quantum coherence for biological function is the pivotal, yet elusive,  question driving research on quantum effects in biological systems. The field promises to bring about a paradigm change in our
understanding of the microscopic mechanisms supporting primary life process on Earth, while expanding our knowledge of, and ability to manipulate, the variety of quantum phenomena that may occur within complex molecular systems that operate at the quantum-classical interface.  To advance this field it is of paramount importance to have clear theoretical understanding of both the quantum coherent behaviours that can be present in prototype biomolecular systems and the relevance of such non-classical features for biological function.

Our work focuses on providing theoretical evidence for a clear and robust relationship between system-environment quantum coherence\textemdash regardless of the basis in which is measured\textemdash and the magnetosensitivity of an avian-inspired quantum compass. Besides demonstrating the necessity and, at times, sufficiency of such collective system and environment quantum coherence for the optimal  detection of magnetic field direction, this work stresses the essential role of environmental degrees of freedom for enabling quantum effects in biomolecular function.

Behavioural studies have shown that migratory birds can orientate via the Earth's magnetic field  \citep{Ritz2009,Wiltschko2005,Wiltschko2011,Wiltschko1993,Beason1984,Beason1995}. However, the physical source of the internal magnetic sensor of these birds is still unclear \citep{Wiltschko2001, Hore2016}. Currently, one leading hypothesis for such internal magnetic compass is the radical-pair mechanism    \citep{Schulten1978,Ritz2000,Hore2016} which can occur in cryptochrome proteins \citep{Wiltschko1972,Schulten1978,Ritz2004,Ritz2009,Kominis2009,Gauger2011,Bandyopadhyay2012,Tiersch2012,Cai2013,Hore2016,Hiscock2016,Worster2016,Kattnig2017,Kattnig2017a,Keens2018,Fay2019}. 
Spin-correlated radical pairs are formed under photo-induced electron transfer. Hyperfine and external magnetic fields---such as the Earth's magnetic field---affect the states of the radical pairs, which in turn affects the subsequent measurable product yields of the recombination reactions \citep{Schulten1978,Ritz2000,Hore2016}. The magneto-sensitive behaviour of radical-pairs has been demonstrated experimentally  \citep{Kerpal2019,Henbest2004, Maeda2008}. However, it is unclear how cryptochrome can meaningfully detect \emph{tiny} magnetic fields, as behavioural studies investigating avian sensitivity to oscillating magnetic fields suggest that birds can detect small changes (of the order of $\unit[10]{n T}$) in an already small Earth's magnetic field ($\unit[50]{\mu T}$) \citep{Ritz2009, Engels2014}. One possibility is that nonclassical phenomena can aid this performance \citep{Hiscock2016}.

Previous works have clearly shown that non-trivial quantum features such as entanglement are present on relevant timescales for the spin dynamics in the radical-pair based magnetic sensor \citep{Cai2013,Gauger2011,Bandyopadhyay2012,Cai2010,Cai2012,Tiersch2012,Hogben2012,Pauls2013,Zhang2014,Zhang2015,Kominis2009,Hiscock2016, Guo2017}. But the necessity of such non-classical correlations for magnetic sensitivity is unclear, with some authors in favour for the necessity of quantum effects \citep{Pauls2013,Cai2013}, and other authors arguing entanglement is just a byproduct \citep{Atkins2019, Gauger2011}. Semiclassical studies on the radical pair magnetic sensor \citep{Fay2019} claim that quantumness overall is unnecessary because they can also predict sensitivity albeit not comparable to quantum based predictions. These studies are overall inconclusive in terms of answering whether or not quantum features contribute to the function of the radical-pair magnetic sensor.

Moreso than entanglement, however, a defining feature of quantum behaviour of a system is quantum coherence \cite{Streltsov2017}. Thus a clear understanding of whether quantum coherence can aid magnetic sensitivity is a central open question. Equally significant yet not widely explored is the role of the nuclei environment in the sensing function of the radical pair. A previous work \cite{Cai2013} has put forward the relevance of global system-environment coherence for the sensitivity of a magnetic compass with the caveat that the environment considered in such study is not widely accepted as representative of the biological conditions of an avian compass.  There, they consider coherence relative to the basis of the system-environment Hamiltonian in the absence of any external magnetic field, and found numerically that global coherence over time played a larger role than local coherence in magnetosensitivity.

Here we consider an avian-inspired quantum magnetic sensor composed of two radicals with a third ``scavenger'' radical \citep{Kattnig2017,Kattnig2017a,Keens2018} under the influence of a collisional environment that can capture a variety of possible decoherence processes \citep{Pezzutto2019,Ziman2010,Cakmak2017,Ciccarello2013,Ciccarello2013a,Lorenzo2017,McCloskey2014,Barra2015,Attal2006}. By characterising the global quantum coherence of our bio-inspired magnetic sensor in a basis-independent manner \citep{Ma2019,Radhakrishnan2019,Streltsov2018} with a measure of the distance to a maximally mixed state, we demonstrate a robust and necessary relationship between system-environment quantum coherence and optimal magnetosensing performance. In doing so, we stress the  essential role of the environment in configuring the relevant quantum behaviour that turns out to be necessary, and in situations, sufficient, for optimal and robust performance in biologically relevant conditions. We discuss how non-maximally mixed system-environment states are likely to be ubiquitous in biomolecular scenarios thereby suggesting that global coherence, irrespective of the basis considered, is likely to be a quantum resource relevant for biomolecular systems.


\section{Results}

\subsection{Basis-independent quantum coherence \label{sec:Coherence}}

Quantum coherence is the most general of all quantum correlations and is one of the key features that distinguish quantum mechanics from classical phenomena \citep{Yao2013, Streltsov2017}. Rigorous frameworks on the resourcefulness of coherence have been constructed \citep{Streltsov2017,Winter2016,Baumgratz2014}, showing how quantum coherence enhances tasks such as quantum metrology and quantum transport on quantum networks \citep{Walschaers2016}.

Generally, coherence is defined relative to some basis: given a preferred basis $\left\{ \ket{i}\right\} _{i}$, incoherent states have form $\sum_{i}p_{i}\ket{i}\bra{i}$ and all other states are deemed coherent. In realistic situations, there are preferred bases such as localised position or photon number. A measure of basis-dependent coherence is given by the $\ell_{1}$ norm \citep{Baumgratz2014}:
$C_{\ell_{1}}\left(\rho\right)=\sum_{i\neq j}\left|\braket{i|\rho|j}\right|$,
which picks out all ``off-diagonal'' terms from the density matrix of a quantum state.  However, this means that basis-dependent coherence is not universal---changing into a system's diagonal basis causes coherence to vanish. Such basis-dependence of this specific form of quantifying coherence has ignited controversy about the relevance and role of quantum coherence in biological systems \cite{Tomasi2020}.

The problem of the choice of basis can be neutralised by considering basis-\emph{independent} quantum coherence \citep{Ma2019,Radhakrishnan2019,Yao2016,Yao2015,Streltsov2018}. If a quantum state has coherence in some basis, then it has basis-independent coherence. The only state without any coherence is the maximally mixed state, $\id_{d}/d=\sum_{i=0}^{d-1}\left(1/d\right)\ket{i}\bra{i}$, which is the same regardless basis chosen.  We can think of basis-independent coherence as a maximisation of coherence over all possible bases. Maximal coherence is equivalent to state purity \cite{Streltsov2018} and  has applications in randomness \cite{Ma2019}. Basis-independent coherence can also be used to describe the degree of polarisation coherence of optical fields \cite{Setaelae2002,Luis2007,Yao2016}.

The amount of basis-independent coherence can be measured with the relative entropy distance to the maximally mixed state \cite{Ma2019}:
\begin{equation}
C_{\id}\left(\rho\right)=S\left(\rho\bigl|\bigr|\id_{d}/d\right)=\log_{2}d-S\left(\rho\right),
\end{equation}
where $S\left(\rho\right)=-\tr\rho\left[\log_{2}\rho\right]$ is the von Neumann entropy and $S\left(\rho||\sigma\right)=-\tr\left[\rho\log\sigma\right]-S\left(\rho\right)$ is the relative entropy. Since basis-independent coherence is equivalent to maximal coherence, $C_{\id}(\rho)$ is also called the relative entropy of purity \cite{Streltsov2018}. 

Note that the term \emph{coherence} is a better descriptor of the underlying quantumness rather than \emph{purity}. In particular, the use of basis-independent coherence is relevant in situations where there is an unstable preferred basis, as is the case here with the quantum magnetic sensor where the magnetic field can change direction. By considering basis-independent coherence, we can robustly quantify the resourcefulness of quantum coherence \emph{without} needing to discuss some preferred basis. Its successful application here gives a precedent to characterising quantum effects in further quantum biological systems. All past works examining basis-dependent coherence remain valid under this approach, while basis-independent coherence further suggests that there could be other sources of coherence (in different basis) that could also contribute to the operation of a quantum system.

In what follows, we often call basis-independent coherence simply as ``coherence'' for succinctness.


\subsection{Avian-inspired quantum magnetic sensor \label{sec:Quantum-magnetic-sensor}}

A source of anisotropy is key in the operation of the radical-pair magnetic sensor \cite{Hore2016}. Typically, this can be derived from hyperfine coupling \citep{Schulten1978,Lee2014,Hore2016,Hiscock2016} or different $g$-factors \citep{Gauger2011,Veselov2001}. Recently, Refs. \citep{Kattnig2017,Kattnig2017a,Keens2018} considered a RPM mechanism with a third ``external scavenger'' particle that used an electron-electron dipolar interaction between radicals to create anisotropy. Here, we introduce anisotropy through system-environment interactions.

We consider a sensor system comprised of three radicals $\mathcal{S}=\left(A,B,C\right)$ which interact via their spin states. While a three-spin radical pair mechanism (RPM) \citep{Kattnig2017,Kattnig2017a,Keens2018} has internal electron-electron dipolar interactions and  electron-exchange interactions, we consider only internal electron-exchange interactions. This is to isolate and highlight the effect of the system-environment interaction on the performance of the magnetic sensor, while also showing that the dipolar interactions are not necessary.
Hence, internal Hamiltonian of the sensor we consider consists of the electron-exchange interaction $\hat{H}_{0}=\hat{H}_{ex}$:
\begin{align}
\hat{H}_{ex} & =-\sum_{\alpha<\beta}J_{\alpha \beta}\left(\dfrac{1}{2}+2\boldsymbol{\hat{S}}^{\alpha}\cdot\boldsymbol{\hat{S}}^{\beta}\right),\label{eq:Hex}
\end{align}
where $\alpha,\beta = A, B,C$ are the system radicals, $\boldsymbol{\hat{S}}^\alpha=\left(S^{\alpha}_{x},S^{\alpha}_{y},S^{\alpha}_{z}\right)=\frac{\hbar}{2}\boldsymbol{\hat{\sigma}}$ are the spin angular momentum operators and $\boldsymbol{\hat{\sigma}}=\left(\sigma_{x},\sigma_{y},\sigma_{z}\right)$ are the Pauli spin operators. Note that $\hbar=1$ is taken throughout this paper.

The quantum magnetic sensor is affected by an external magnetic field
\begin{equation}
\vec{\boldsymbol{B}}_{0}\left(\theta,\phi\right)=B_{0}\left(\cos\theta\sin\phi,\sin\theta\sin\phi,\cos\phi\right),
\end{equation}
via the Zeeman interaction:
\begin{align}
\hat{H}_{\vec{\boldsymbol{B}}}\left[\theta,\phi\right] & =\gamma\vec{\boldsymbol{B}}_{0}\cdot\sum_{\alpha=A,B,C}\boldsymbol{\hat{S}}^{\alpha},\quad\gamma\equiv\dfrac{g\mu_{B}}{\hbar},
\end{align}
where $\phi\in\left[0,\pi\right]$ and $\theta\in\left[0,2\pi\right)$, $\mu_{B}$ is the Bohr magneton, and $g\approx2$. Note that the magnetic field does not affect the environment in our model.

The complex protein environment of the radical pair mechanism in animals has not been fully experimentally characterised nor decided \citep{Solovyov2007, Hore2016, Liedvogel2009}. 
As such, it is important to chose a flexible model framework.
We have chosen to use a collisional model for our environment because collisional models can simulate and represent a wide range of environment models, both Markovian and non-Markovian \citep{Scarani2002,Pezzutto2019,Ziman2010,Cakmak2017,Ciccarello2013,Ciccarello2013a,Lorenzo2017,McCloskey2014,Barra2015,Attal2006,Rybar2012, Amato2019}, which allows our framework to easily extend to more general situations and describe the unavoidable system-environment interactions in biological systems.

We model the environment as three separate baths of two-level particles that interact with the quantum sensor radicals. For the purpose of illustration, we consider a system-environment interaction which causes the system radical spins to decohere. In particular, we consider a ``controlled-NOT (CNOT)''-like interaction which induces the standard Lindblad decoherence provided the environment initial states are maximally mixed. This interaction furthermore promotes the spread of system information into the environment \citep{Balaneskovic2015}.  The CNOT interaction is a `worse case' scenario: if the system information is perfectly copied into the environment, the system will become classical and objective \citep{Zurek2009} thus affecting the operation of the sensor. This allows us to see ``how close'' to the quantum-classical boundary the sensor can operate.

We also suppose that the environment spins have no internal interactions, and we assume their initial states are all identical.  The system-environment interaction proceeds in two steps that are repeated, such that each collisional iteration $\#n$ takes time $T=\tau_{se}+\tau_{ee}$. This is shown in Fig.   \ref{fig:Collisional/repeated-interaction}.
\begin{enumerate}
\item For a time $\tau_{se}$ period, environment particles $\left(\mathcal{E}_{n,A},\mathcal{E}_{n,B},\mathcal{E}_{n,C}\right)$ will collide with system radicals $A$, $B$, and $C$ separately with the following controlled-NOT-like (CNOT) interaction:
\begin{equation}
\hat{V}_{\mathcal{SE}}\left(t\right)=\dfrac{J_{se}}{2}\sum_{\alpha=A,B,C}\left(\id -\sigma_z\right)^{\alpha}\otimes\left(\id-\sigma^{\mathcal{E}_{n,\alpha}}_{x}\right),\label{eq:V_SE}
\end{equation}
where $t\in\left[nT,nT+\tau_{se}\right)$ for $n=0,1,\ldots$. Note that if $J_{se}\tau_{se}=\pi/2$, then we would have the CNOT interaction between the system and environment if there were no other interactions:  $e^{-i\hat{V}_{\mathcal{SE}}\pi/2}=U_{\text{CNOT}}$.
\item Then, for a time $\tau_{ee}$ period, the system continues to evolve but there is no system-environment interaction, i.e., $\hat{V}_{\mathcal{SE}}\left(t\right)=0$.  This represents a constant time between environment collisions.\footnote{It is possible to consider probabilistic times between successive environment collisions \citep{Garcia-Perez2020}. It is also possible to implement a collision model with memory, in which the $\#n$ environment spins interact with the future $\#n+1$ environment spins \citep{Pezzutto2019}.}
\end{enumerate}

\begin{figure}
\begin{centering}
\includegraphics[width=1\columnwidth]{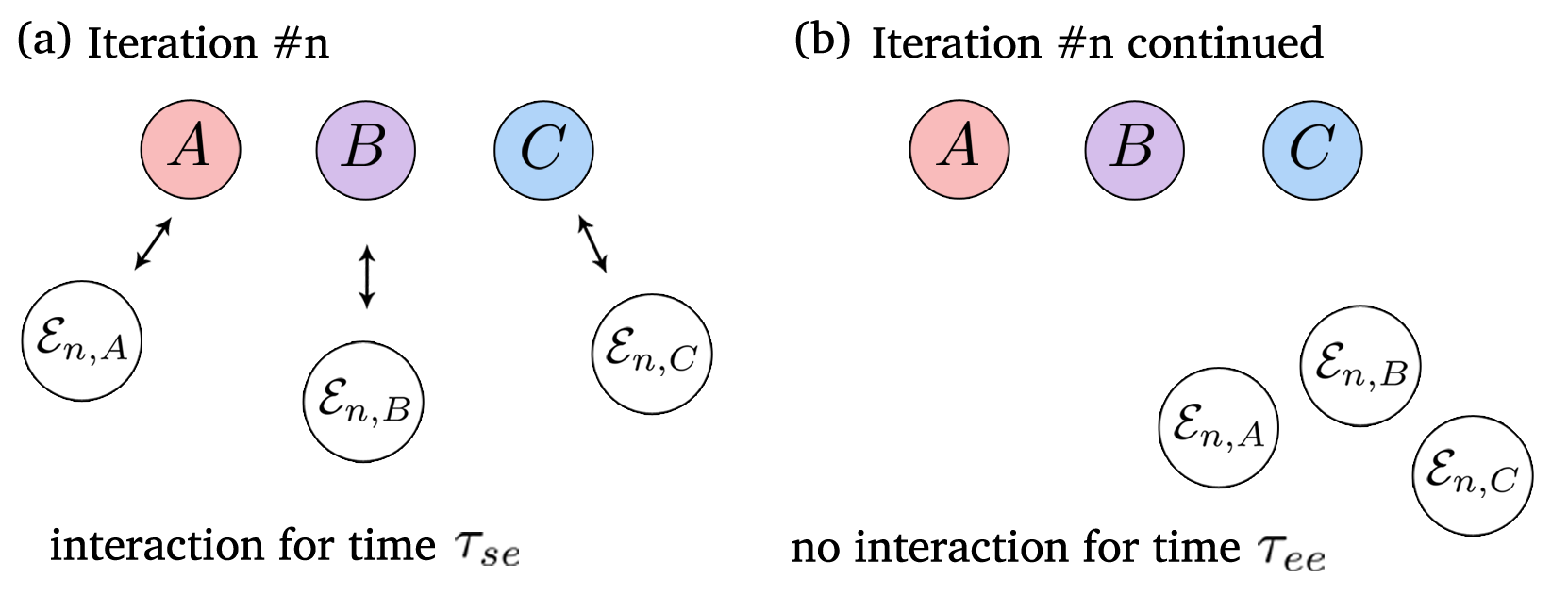}
\par\end{centering}
\caption{Collisional model of the environment. The quantum magnetic sensor consists of radical spins $A$, $B$ and $C$. At iteration $\#n$, \textbf{(a) }three environment particles (also two-level) interact separately with the radicals for time $\tau_{se}$. \textbf{(b) }Then for a period of $\tau_{ee}$ there is no interaction. A new set of environment particles interacts with the systems in iteration $\#n+1$.\label{fig:Collisional/repeated-interaction}}
\end{figure}

The complete evolution of the system and environment is given by \citep{Keens2018}:
\begin{equation}
\dfrac{d\rho_{\mathcal{SE}}\left(t\right)}{dt}=-i\left[\hat{H}_{total}(t),\rho_{\mathcal{SE}}\left(t\right)\right]-k\rho_{\mathcal{SE}}\left(t\right),\label{eq:evolution}
\end{equation}
where $\hat{H}_{total}\left(t\right)=\hat{H}_{0}+\hat{H}_{\vec{\boldsymbol{B}}}\left[\theta,\phi\right]+\hat{V}_{\mathcal{SE}}\left(t\right)$, and  $k$ is the expected lifetime of the recombination. This is disjoint for the different collisional iteration. The recursive solution to the evolution is given in the \hyperref[sec:methods]{Methods}.

In this paper, we consider the radicals to lie on the $z$-axis, $A-B-C$ with equal spacing between $A-B$ and $B-C$, i.e. $A:\left(0,0,-0.5\right)$, $B:\left(0,0,0.5\right)$ and $C:\left(0,0,1.5\right)$, in implicit units of $\left(d_{0}/d_{AB}\right)^{\frac{1}{3}}$. Other parameters are given in Table \ref{tab:Parameters-and-configurations}, with the magnetic field strength equal to the Earth's magnetic field $B_{0}=\unit[50]{\mu T}$, and with an expected recombination decay lifetime of $k=\unit[1]{\mu s}$ \citep{Hore2016}. All initial states (system radicals and environment particles) are product, and hence have no initial correlations.

\begin{table}
\caption{Parameters and configurations used in the paper (unless otherwise specified). The system $A,B,C$ set in a straight line on the $z$ axis, space apart as: $z_{A}=-0.5$ $z_{B}=0.5$ and $z_{C}=1.5$, in implicit units of $\left(d_{0}/d_{AB}\right)^{\frac{1}{3}}$ and with $d_{AB}=\unit[20]{\text{{\normalfont\AA}}}$ ($\hbar=1$). S-E refers to system-environment. \label{tab:Parameters-and-configurations}}

\begin{centering}
\begin{tabular}{|>{\centering}m{3.2cm}|>{\centering}m{4.8cm}|}
\hline 
Parameters & \multicolumn{1}{c|}{Values} \tabularnewline
\hline 
\hline 
dipolar exchange $\hat{H}_{dd}$ & N/A \tabularnewline
\hline 
electron-exchange $\hat{H}_{ex}$ & isotropic $J_{ABC}/d_{AB}=1$ \tabularnewline
\hline 
decay lifetime  & $k/d_{AB}=0.0245$ ($k=\unit[1]{\mu s}$)\tabularnewline
\hline 
magnetic field & $\gamma B_{0}/d_{AB}=0.215$ ($B_{0}=\unit[50]{\mu T}$)\tabularnewline
\hline 
initial $\rho_{n,\mathcal{E}}\left(0\right)$ & $\id/2$\tabularnewline
\hline 
S-E interaction time & $d_{AB}\tau_{se}=1.0$ \tabularnewline
\hline 
S-E interaction strength & $J_{se}\tau_{se}=\dfrac{\pi}{2}$\tabularnewline
\hline 
time between collisions & \centering{}$d_{AB}\tau_{ee}=1.0$\tabularnewline
\hline 
\end{tabular}
\par\end{centering}
\centering{}%
\end{table}

\subsection{Performance measures}

Over time, the radical pairs $A,B$ chemically recombine. The quantum yield of this singlet recombination product is:
\begin{equation}
\varphi_{\text{singlet}}\left[\theta,\phi\right]=k\int_{0}^{\infty}d\tau\tr\left[P_{\text{singlet}}^{\left(A,B\right)}\rho_{\mathcal{SE}}\left(\tau|\theta,\phi \right)\right],\label{eq:quantum_yield}
\end{equation}
where the evolution of the state $\rho_{\mathcal{SE}}$ depends on the magnetic field directions $(\theta,\phi)$, and $P_{\text{singlet}}^{\left(A,B\right)}$ is a singlet state in primary radical pair $A$ and $B$:
\begin{equation}
P_{\text{singlet}}^{\left(A,B\right)}=\dfrac{1}{4}\left(\id-\boldsymbol{\hat{\sigma}}^{A}\cdot\boldsymbol{\hat{\sigma}}^{B}\right).
\end{equation}
Different field directions causes the recombination value $\varphi_{\text{singlet}}\left[\theta,\phi\right]$ to change: by detecting these changes, we can in turn infer the possible field direction. It is not known exactly how the difference in recombination yield could be perceived \emph{in vivo}. We will be considering primarily the \emph{absolute anisotropy} of the yield \citep{Kattnig2017}
\begin{align}
\Delta\left[\varphi_{\text{singlet}}\right] & =\max_{\left(\theta,\phi\right)}\varphi_{\text{singlet}}\left[\theta,\phi\right]-\min_{\left(\theta,\phi\right)}\varphi_{\text{singlet}}\left[\theta,\phi\right].
\end{align}

We find that the results for the \emph{relative anisotropy} tend to be qualitatively similar,
\begin{align}
RA\left[\varphi_{\text{singlet}}\right] & =\Delta\left[\varphi_{\text{singlet}}\right]/\left\langle \varphi_{\text{singlet}}\right\rangle ,
\end{align}
where the orientation-averaged singlet yield is
\begin{equation}
\left\langle \varphi_{\text{singlet}}\right\rangle =\dfrac{1}{4\pi}\int_{0}^{2\pi}d\theta\int_{0}^{\pi}d\phi\sin\left(\phi\right)\varphi_{\text{singlet}}\left[\theta,\phi\right].
\end{equation}

We are primarily interested in the effect of the \emph{initial} coherence and how it affects subsequent evolution. 
However it is also possible to quantify the contribution of coherence over the time period where recombination occurs. To do so we can consider \emph{coherence yields}, akin to the recombination yield and the entanglement yield of Ref. \citep{Fassioli2010}, and coherence-like yields in prior works on the avian magnetic sensor \citep{Cai2010}. The coherence yield is an integration of the coherence $C\left(t\right)$ over time, weighted by the integrand factor $\tr[P_{\text{singlet}}^{\left(A,B\right)}\rho_{\mathcal{S}}\left(\tau\right)]$ of the recombination yield:
\begin{equation}
C_{\text{yield}}\left[\theta,\phi\right]=\dfrac{k}{\varphi_{\text{singlet}}}\int_{0}^{\infty}d\tau  C\left(\tau\right)\tr[P_{\text{singlet}}^{\left(A,B\right)}\rho_{\mathcal{SE}}\left(\tau\right)].\label{eq:coherence_yield}
\end{equation}
The absolute anisotropy yield of the coherence \emph{etc.} can be defined analogously.

\subsection{Initial basis-independent coherence is necessary for magnetic sensing}

In this section, we show that initial basis-independent coherence is necessary for a nontrivial magnetic sensor. We proceed by first by considering the crucial Zeeman magnetic interaction and examining the coherence required for a nontrivial evolution. We then show how basis-independent system-environment coherence is necessary when the evolution of the system-environment is unital.

The Zeeman interactions between the magnetic sensor and the geomagnetic field, $\hat{H}_{\vec{\boldsymbol{B}}}$, is crucial for the performance of the magnetic sensor. A necessary condition for the system state $\rho_{\mathcal{S}} (0)$ to evolve nontrivially under the Zeeman interaction is that $[\rho_{\mathcal{S}}(0), \hat{H}_{\vec{\boldsymbol{B}}} (\theta,\phi)]\neq 0$ for at least one set of angles $(\theta,\phi)$. The following Lemma gives all the states are \emph{trivial} under the Zeeman interaction:

\begin{lem}
$\left[\rho_{\mathcal{S}}\left(0\right),\hat{H}_{\vec{\boldsymbol{B}}}(\theta,\phi)\right]=0$ for all angles $(\theta,\phi)$ if and only if
\begin{align}
\rho_{\mathcal{S}} =&\dfrac{\id_{ABC}}{8}+p_{AB}\vec{\sigma}^{A}\cdot\vec{\sigma}^{B}+p_{AC}\vec{\sigma}^{A}\cdot\vec{\sigma}^{C}+p_{BC}\vec{\sigma}^{B}\cdot\vec{\sigma}^{C} \nonumber \\
&+p_{ABC}(\sigma_{x}\otimes\sigma_{y}\otimes\sigma_{z}+\sigma_{z}\otimes\sigma_{x}\otimes\sigma_{y}  \nonumber\\
&\qquad\qquad  +\sigma_{y}\otimes\sigma_{z}\otimes\sigma_{x} - \sigma_{x}\otimes\sigma_{z}\otimes\sigma_{y} \nonumber\\
&\qquad\qquad  -\sigma_{y}\otimes\sigma_{x}\otimes\sigma_{z}  -\sigma_{z}\otimes\sigma_{y}\otimes\sigma_{x}).\label{eq:reduced_rho_S_HB_only}
\end{align}
\label{lemma:reduced_rho_S_HB_only}
\end{lem}

The proof is given in the Appendix \ref{sec:Lemmarho_HB}. It proceeeds by writing the general system state as $\rho_{\mathcal{S}} = \sum_{ijk=0,x,y,z} p_{ijk} \sigma_i^A \otimes \sigma_j^B \otimes \sigma_k^C$ and enforces the zero commutation condition for various field directions of $\hat{H}_{\vec{\boldsymbol{B}}}$ until the reduced for satisfies the commutation for all angles.

Lemma~\ref{lemma:reduced_rho_S_HB_only} says that initial basis-independent coherence is necessary but not sufficient, as there exists trivial states in the lemma that contain correlations between systems. However, note that the trivial states all have zero local coherence, \emph{i.e.} their reduced states are maximally mixed $\rho_A=\rho_B=\rho_C=\id/2$. Hence, having coherent local states is sufficient (but not necessary) for the system to evolve nontrivially under the geomagnetic field.

However, the Zeeman interaction alone is not sufficient for a magnetic sensor. Since $\hat{H}_{\vec{\boldsymbol{B}}} (\theta,\phi)$ commutes with the singlet state $P^{(A,B)}_{\text{singlet}}$ regardless of field angle, the integrand of the recombination yield from Eq.~(\ref{eq:quantum_yield}) takes the same value if we only have  the Zeeman interaction:
\begin{align}
 &\tr[P^{(A,B)}_{\text{singlet}} \rho_\mathcal{SE}(t)] \nonumber \\
&=e^{-kt} \tr[P^{(A,B)}_{\text{singlet}} e^{-i \hat{H}_{\vec{\boldsymbol{B}}} t}\rho_{\mathcal{SE}}(0)  e^{i \hat{H}_{\vec{\boldsymbol{B}}} t}] \label{eq:HB_trace_only_1}\\
&=e^{-kt} \tr[P^{(A,B)}_{\text{singlet}}\rho_{\mathcal{SE}}(0)],\label{eq:HB_trace_only_2}
\end{align}
rendering the same recombination yield value $\varphi_\text{singlet}$, regardless of field direction $ (\theta,\phi)$. System-system and/or system-environment interactions are needed to break the symmetry of the magnetic field acting on the sensor. In the next section, we will consider the detailed effects of such interactions. For now, in the last part of this section, we show that basis-independent system-coherence can be necessary for a sensor.

Any radical-pair magnetic sensor that has an evolution equation described by Eq.~(\ref{eq:evolution}), regardless of the details of the Hamiltonian $\hat{H}_{total}$, has \emph{unital evolution}, which means that initially maximally mixed states \emph{remain maximally mixed} (up to trace). Because of this, nontrivial evolution requires basis-independent coherence.

\begin{lem}
A radical-pair magnetic sensor described by a unital evolution that is nontrivial must have initial states that are not $\rho_\mathcal{SE}(0) \neq \id_{d_\mathcal{SE}}/d_\mathcal{SE}$, where $d_\mathcal{SE}$ is the dimensions of the combined system-environment.\label{lem:RPM_unital_coherence}
\end{lem}

The proof is given in the \hyperref[sec:methods]{Methods}.

The initial system-environment state should be non-maximally mixed state for nontrivial performance given a unital model.  The necessity of initial coherence in the system-environment state is due to the unital nature of the evolution: if the initial state is maximally mixed, unital evolution keeps the maximally mixed state and does not generate any new coherence. Hence, coherence must be `injected' into the radical-pair at the beginning.

\subsection{The necessity and sufficiency of basis-independent system-environment coherence and non-trivial system-environment interactions}

Quantum biological systems are not closed: thus, their environmental interactions and influences are unavoidable and play a nontrivial role in the evolution of biological systems.

As found in the previous section, although the geomagnetic field interaction is crucial, there must also be other interactions in order to have a nontrivial magnetic sensor. These other interactions must introduce anisotropy, breaking the symmetry of the Zeeman interaction \cite{Hore2016}.   In past works, this is typically done with anisotropic hyperfine interactions with nuclear spins \cite{Hore2016}, or with a dipolar interaction with third scavenger \cite{Kattnig2017,Keens2018}. Here, we do this with a system-environment interaction $\hat{V}_{\mathcal{SE}}$.

The quantum magnetic sensor requires the system to affected differently under different field angles, and that furthermore these changes must be detectable in the recombination yield of Eq.~(\ref{eq:quantum_yield}), which recall has integrand $\tr[P_{\text{singlet}}^{\left(A,B\right)}\rho_{\mathcal{S}}\left(t\right)]$. As such, noncommutativity of the evolution Hamiltonians with the system state \emph{and} with the singlet state is necessary.

\begin{lem}
For nontrivial performance, it is necessary that the following hold:
\begin{align}
&[\rho_{\mathcal{SE}}\left(0\right),\hat{H}_0+\hat{V}_{\mathcal{SE}}+\hat{H}_{\vec{\boldsymbol{B}}}(\theta,\phi)]\neq0, \\
&[P_{\text{singlet}}^{\left(A,B\right)},\hat{H}_0+\hat{V}_{\mathcal{SE}}]\neq0, \label{eq:necessary_non_trivial_P_H_V}, \\
&[\hat{H}_{\vec{\boldsymbol{B}}} (\theta,\phi),\hat{H}_0+\hat{V}_{\mathcal{SE}}]\neq 0,
\end{align}
for at least one angle $(\theta,\phi)$, where $\hat{H}_0$ is some internal system evolution. \label{lemma:necessary}
\end{lem}

The proof is in the \hyperref[sec:methods]{Methods}. Each statement is a necessary condition for the system-environment state to evolve nontrivially and a necessary condition for the singlet recombination to change.

For our particular model, initial basis-independent coherence is necessary and sufficient for non-commutativity of the various commutators:

\begin{proposition}
For the model we consider [Sec. \ref{sec:Quantum-magnetic-sensor} with parameters in Tab.~\ref{tab:Parameters-and-configurations}, i.e. with initial environment state $\rho_{\mathcal{E}}=\id_{\mathcal{E}}/d_{\mathcal{E}}$, $\hat{V}_{\mathcal{SE}}$ in Eq.~(\ref{eq:V_SE}) and $\hat{H}_0 = \hat{H}_{ex}$ in Eq.~(\ref{eq:Hex}), $J_{\alpha,\beta}=J_{ABC}$ for all $\alpha,\beta =A,B,C$],
the terms of Lemma \ref{lemma:necessary} hold if and only if the initial state on the system is not maximally mixed $\rho_{\mathcal{S}} \neq \id/d_\mathcal{S}$. \label{prop:necessary_conditions_hold}
\end{proposition}

\begin{proof}
The three components of Lemma \ref{lemma:necessary} are the commutators:
\begin{enumerate}
    \item $[\rho_{\mathcal{S}}\left(0\right)\otimes \id_{\mathcal{E}} / d_{\mathcal{E}},\hat{H}_0+\hat{V}_{\mathcal{SE}}+\hat{H}_{\vec{\boldsymbol{B}}}(\theta,\phi)]\neq0$ : see Proposition \ref{prop:rho_S_H_our_model}, which we prove for the particular model we consider.
    \item $[P_{\text{singlet}}^{\left(A,B\right)},\hat{H}_0+\hat{V}_{\mathcal{SE}}]\neq0$ : see Lemma~\ref{lemma:necessary-VSE} and Lemma \ref{lem:P_AB_H_VSE} where we prove that this holds whenever $\hat{V}_{\mathcal{SE}}$ has any nontrivial system-environment interactions.
    \item $[\hat{H}_{\vec{\boldsymbol{B}}} (\theta,\phi),\hat{H}_0+\hat{V}_{\mathcal{SE}}]\neq 0$ : see Lemma \ref{lem:H_B_H_0_V} which also holds for general nontrivial system-environment interactions $\hat{V}_{\mathcal{SE}}$ provided that the magnetic field does not interact with the environment.
\end{enumerate}
\end{proof}

Note that we consider a system configuration \emph{without} the dipolar interaction. It shows that the dipolar interaction is \emph{not} required for the performance of the magnetic sensor, unlike \citep{Keens2018}\textemdash this is due to the environment taking up the role of breaking symmetry.

Furthermore, note that only item (1) of the proof of Proposition~\ref{prop:necessary_conditions_hold} uses the particular form for the various interactions Hamiltonians. While it remains open whether this holds other radical-pair mechanism models, past work suggests that it will \cite{Cai2013}.

Prompted by our numerical results later in the paper (Fig.~\ref{fig:Plot-of-initial-vs-yield-random}), we propose the following conjecture, which we believe is likely to be true, but we have not proven:

\begin{conjecture}
The necessary conditions in Lemma \ref{lemma:necessary} are also sufficient for nontrivial sensing.\label{conj:sufficient}
\end{conjecture}

Combining Conjecture \ref{conj:sufficient} with  Proposition \ref{prop:necessary_conditions_hold}, the following may also be true, though once again unproven:

\begin{conjecture}
For the model we consider  (Sec. \ref{sec:Quantum-magnetic-sensor} with parameters in Tab.~\ref{tab:Parameters-and-configurations}), a sufficient condition for nontrivial sensing, is $\rho_{\mathcal{S}}\left(0\right)\neq\id_{\mathcal{S}}/d_{\mathcal{S}}$.
\end{conjecture}

\begin{figure*}
\begin{centering}
\includegraphics[width=1\textwidth]{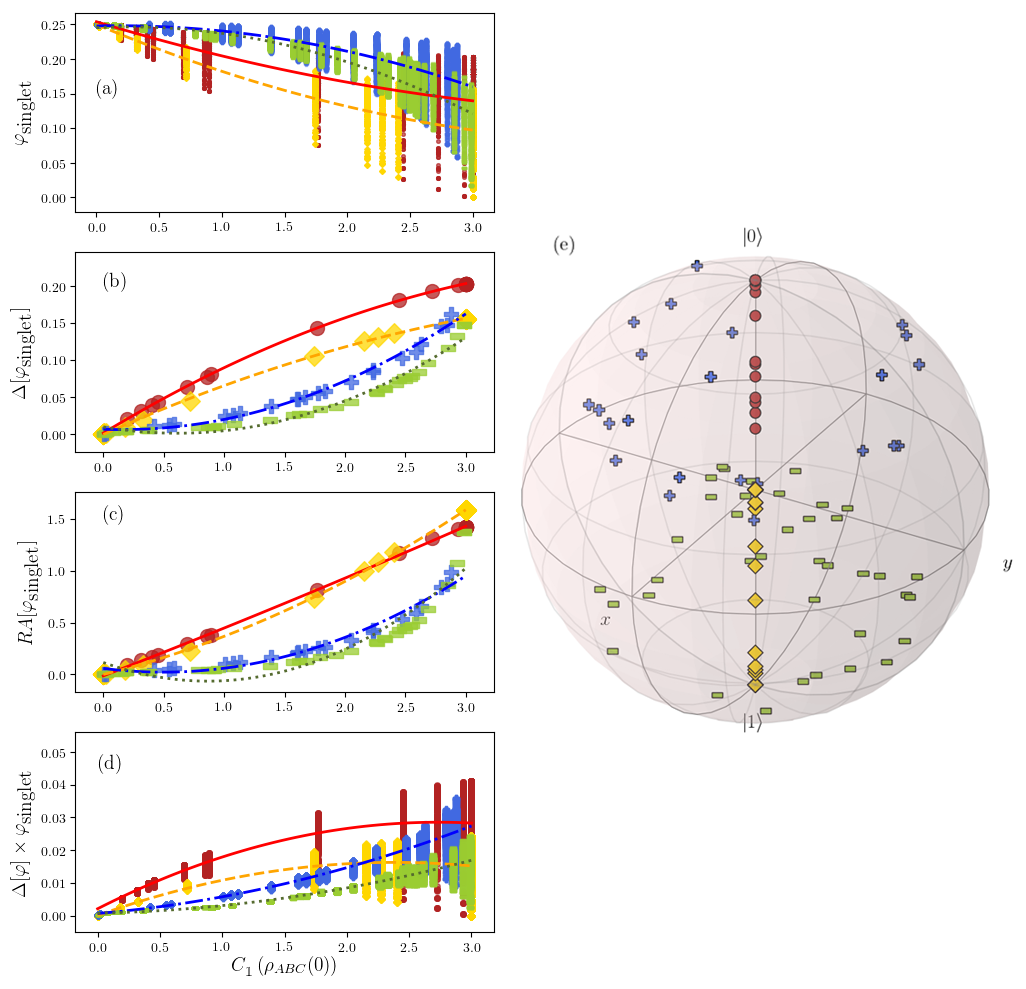}
\par\end{centering}
\caption{ Plots of initial system state's basis-independent coherence $C_{\id}(\rho_{ABC}(0))$ with the (a) recombination yield $\varphi_{\text{singlet}}$, (b) absolute anisotropy of the yield $\Delta[\varphi_{\text{singlet}}]$, (c) relative anisotropy of the yield $RA[\varphi_{\text{singlet}}] = \Delta[\varphi_{\text{singlet}}]/\langle\varphi_{\text{singlet}} \rangle$, and (d) product of the absolute anisotropy of the yield with the yield $\Delta[\varphi_{\text{singlet}}] \times \varphi_{\text{singlet}}$. (e) The initial state $\rho_0$ is randomly chosen (150 different random initial states), and displayed here on the Bloch sphere. The complete initial system state is $\rho_{ABC} (0) =\rho_0\otimes \rho_0 \otimes \id^C /2$
Model parameters are listed in Table \ref{tab:Parameters-and-configurations}.\label{fig:Plot-of-initial-vs-yield-random}}
\end{figure*}

Fig.~\ref{fig:Plot-of-initial-vs-yield-random} supports our conjecture that nonzero initial coherence is sufficient for a nontrivial sensor. In Fig. \ref{fig:Plot-of-initial-vs-yield-random}, we see that there is a general trend between greater initial coherence and large range of recombination yield $\varphi_{\text{singlet}}$ which is integral to a sensitive magnetic sensor.  There are four sub-trends that we identify with different colours and shapes: initial states diagonal in $\sigma_z$ versus those not, and initial states with positive versus negative $z$ coordinate (given by $\tr[\sigma_z \rho_A(0)]$). The extra dependence on $z$ is due to basis-dependence of the system-environment interaction $\hat{V}_\mathcal{SE}$, while the sign dependence is due to the sign of the interaction. This can be seen clearly in the Bloch sphere depicting the initial state on $\rho_A(0)$.

In a realistic situation, we must balance both the variation in the recombination yield and the actual value of the yield---\emph{i.e.}, while a large varition is ideal, the value of the yield itself needs to be large enough to be feasibly detectable. One measure of this balance is the objective function $\Delta[\varphi_\text{singlet}]\cdot \varphi_\text{singlet}$ \emph{i.e.} the product of the anisotropy with the value of recombination. If this function is large, then we have both good variation and a large overall recombination values. We plot this in  Fig. \ref{fig:Plot-of-initial-vs-yield-random}(d). We see that for states lying on the $z$-axis in particular, that there is an initial rise before a partial plateau starting from $C_\id (\rho_{ABC}(0))\approx 1.5$, which suggests that after a certain point, increasing more initial coherence leads to diminishing returns in the performance of the sensor.

In Appendix \ref{sec:coherence_yield_results}, we also numerically examine the relationship between recombination yield and coherence yield (the contribution of coherence to the yield over time). We find that they tend to be inversely linearly correlated, similar to the inverse correlation between the initial coherence and the recombination yield. Thus, while the increase in basis-independent coherence during the  recombination process suppresses the yield amount, a certain amount of initial coherence is needed to simply have a wide \emph{variation} in the recombination yield. This tradeoff is exemplified back in Fig.~\ref{fig:Plot-of-initial-vs-yield-random}(d).

While these sections have been focused on the initial coherence in the system state, we also examine the situation where the system is initially maximally mixed and the environment has coherence in Appendix~\ref{app:SWAP}, which is then a resource in the performance of that magnetic sensor.

Beyond coherence, the system and environment can also share correlations, which we can use to formally characterise how the system operates at the quantum-to-classical boundary. In Appendix \ref{sec:quantum_Darwinism}, we consider the quantum correlations between the system and environment, through the lens of Quantum Darwinism \cite{Zurek2009,Le2019,Horodecki2015}. Quantum Darwinism which extends the notion of nonclassicality to \emph{nonobjectivity}, where macroscopic objects are objective. We find that the sensor can operate with imperfect classical correlations between system and environment and very little quantum discord. This lack of perfect classical correlations corresponds to a non-objective situation \cite{Le2019, Horodecki2015}. Note that perfect correlations at all times would imply that the environment is constantly, perfectly, monitoring the system\textemdash which is a quantum zeno-like situation. The very low quantum discord also confirms that quantum coherence is the key quantum effect involved in the performance of the avian-inspired quantum magnetic sensor.

\section{Discussion \label{sec:discussion}}

The radical pair mechanism is the strongest contender for explaining high magnetic sensitivity in birds. To understand how that sensitivity is possible, it is crucial for us to investigate the role of quantum coherence in the radical pair mechanism. Here, we provide a concrete demonstration that coherence indeed aids magnetic sensing based on radical pairs.
In the radical-pair model considered, we show that initial system-environment coherence---in \emph{any} basis---is necessary for a nontrivial sensor, and our results suggest that this is sufficient when combined with a symmetry-breaking interaction.

In an avian magnetic sensor, there is no (immediately) preferred basis. The Earth's external magnetic field can be aligned in any direction relative to the magnetic sensor, which we assume is mobile. As the relative magnetic field direction varies, the instantaneous basis changes: hence a state that was once incoherent now appears coherent. As such, consideration of basis-\emph{independent} coherence become important. 

In Ref. \citep{Tomasi2020}, the authors examine the bases of coherence and noise in light-harvesting processes. There, they noted that coherence has an effect if noise processes act on that coherence---which means that basis in which a state has coherence is \emph{different} from the basis in which noise processes act. We see a similar situation in the magnetic sensor: as the external magnetic field can have any relative direction, it can act in any basis, hence coherence in any basis becomes useful.

Our results hence suggest that basis-independent coherence is a quantum resource for the magnetic sensor. In realistic fluctuating biological environments, the initial system radicals are unlikely to be in the maximally mixed state due to fluctuations. Furthermore, the chemical process producing radical pairs in cryptochrome produce singlet and triplet states which are non-maximally mixed states \cite{Hore2016}. Hence, the relevance of coherence for the magnetic sensor performance is robust: the sensor continues to perform under noise, so long as it does not cause the system to become maximally mixed.

Our results also confirm the statements made by the previous work of \citet{Cai2013}, i.e. that global system-environment coherence is crucial for the sensitivity of a magnetic compass. From the perspective of our work, we can understand the results of  \citet{Cai2013} as follows: they start with an initially coherent singlet state, and randomly generate the dynamical Hamiltonians in their unital model. Different random Hamiltonians have different anisotropies and coherence-generation ability, thus giving rise to differing levels of coherence over time and different performance. Note that \citet{Cai2013} consider coherence that is dependent on the basis of their system-environment Hamiltonian and assume an environment that may not have a strong bio-inspired basis. In contrast, our results do not depend on basis, our model has the ability to represent a wide range of environments, and the majority of our results are independent of the specific interaction details.

Note that our model is \emph{unital} (and trace-decreasing). Under this evolution, the maximally mixed state is preserved (up to a real factor) $\Lambda(\id/d) = \lambda \id /d$. The supplemental of \citet{Hogben2012}'s paper contains a sensor that has maximally mixed initial state $\id/4$ on the radicals, $\id/2$ on the single nuclear spin. There, the sensor is nontrivial whenever the singlet and triplet reaction rates are unequal $k_S \neq k_T$---this leads to non-unital dynamics in the master equation, which therefore shifts the resource from the initial states with coherence to the coherence-generating dynamics.

Previously, \citet{Cai2010} had suggested that \emph{entanglement} rather than coherence in the initial state played the greater role, as their optimal magnetic sensitivity for incoherent states (with respect to the standard computational basis) was the same as the magnetic sensitivity for separable initial states. Meanwhile initial entangled system states preformed better than both. However, their definition of magnetic field sensitivity is focused on sensitivity to changing \emph{field strength} rather than changing field angle. Furthermore,  \citet{Cai2010} also showed that incoherent and separable initial states give a nontrivial magnetic sensor\textemdash and this is consistent with our results that initial basis-independent coherence is required. 
We have located the source for magnetic field direction sensitivity: \emph{basis-independent} coherence. Our work then shows that avian magnetic sensors fundamentally require coherence, in any basis, whether it be in the initial state, or during the evolution itself.

The measures of basis-independent coherence correspond to measures of purity, which is more easily experimentally accessible than quantum correlations such as discord and entanglement: there exist schemes that do not require quantum state tomography, especially for the linear purity, $\tr[\rho^2]$ \cite{Nakazato2012, Filip2002, Ekert2002, Islam2015, Pichler2013,Lee2011}.

In general, we show that the magnetic sensor benefiting from coherence needs to have a symmetry breaking interaction---which can be a nuclear hyperfine interaction, or dipolar interaction with scavenger radicals, or a general system-environment interaction as we have shown here. \citet{Keens2018} have argued that the dipolar exchange was necessary, provided there were no hyperfine interactions and no environment. We found that in the presence of environmental interactions, which are ubiquitous in biological systems, the dipolar exchange is not necessary for a magnetic sensor. The added environment helps break the symmetry, analogous to the symmetry breaking performed by the dipolar exchange and hyperfine interactions in previous works. Our work therefore emphasises how integral it is to consider the environment when examining candidates for avian magnetosensing. 

Our results also open the door for further investigations of the role of the environment in the sensitivity of the magnetic sensor. For instance, while we have considered the scenario of a memoryless environment, collisional models can be also used to investigate non-Markovian environments \citep{Pezzutto2019,Cakmak2017,Ciccarello2013a} which could enhance performance further. Coherence in the environment itself could also act as a resource to further enhance performance of bio-inspired sensors  \citep{Rodrigues2019}.

Overall we have shown strong connections between symmetry breaking dynamics, system-environment interactions and initial system-environment coherence, and the subsequent performance of a quantum magnetic sensor based on the radical pair mechanism. Radical pairs play an pivotal role in a variety of biological functions, including, for instance, enzyme catalysis and associated cell metabolic process enabled by coenzyme $B_{12}$ which has displayed interesting magnetic effects \cite{Hughes2019}. Our work therefore provides an important framework to understand the advantages that quantum coherence may provide in a variety of biological and chemical process involving radical pair dynamics \cite{Hore2020}.


\section{Methods \label{sec:methods}}

\subsection{System-environment evolution \label{app:System-environment-evolution}}

The system-environment evolution is composed of a series of interactions between the system and fresh environment states, and the solution is recursive.

Firstly, define the eigenstates and eigenenergies of the total Hamiltonian for times $t\in\left[nT,nT+\tau_{se}\right)$ and $t\in\left[nT+\tau_{se},\left(n+1\right)T\right)$ respectively:
\begin{align}
\hat{H}_{total}\left(t\right) & =\hat{H}_{0}+\hat{H}_{\vec{\boldsymbol{B}}}\left[\theta,\phi\right]+\hat{V}_{\mathcal{SE}}\left(t\right) \nonumber \\
&=\sum_{i}\omega_{i}\ket{\phi_i}\bra{\phi_i}_{\mathcal{SE}_{n}}\\
\hat{H}_{total}^{\prime}\left(t\right) & =\hat{H}_{0}+\hat{H}_{\vec{\boldsymbol{B}}}\left[\theta,\phi\right]=\sum_{\alpha}\xi_{\alpha}\ket{\alpha}\bra{\alpha}_{\mathcal{S}},
\end{align}
leading to the unital operators $\mathbb{U}_{\tau}$ and $\mathbb{W}_{\tau}$:
\begin{align}
\mathbb{U}_{\tau} & \equiv\exp\left[-i\tau\hat{H}_{total}\right]=\sum_{j}e^{-i\omega_j\tau}\ket{\phi_j}\bra{\phi_j},\\
\mathbb{W}_{\tau} & \equiv\exp\left[-i\tau\hat{H}_{total}^{\prime}\right]=\sum_{\alpha}e^{-i\xi_{\alpha}\tau}\ket{\alpha}\bra{\alpha}.
\end{align}

For time $t\in\left[0,\tau_{se}\right)$, the system-environment interacts with Hamiltonian $\hat{H}_{total}=\hat{H}_{\mathcal{SE}_{0}}=\hat{H}_{0}+\hat{H}_{\vec{\boldsymbol{B}}}+\hat{V}_{SE_0}$ on the system environment $\rho_{\mathcal{SE}_{0}}$:
\begin{equation}
\dfrac{d\rho_{\mathcal{SE}_{0}}\left(t\right)}{dt}=-i\left[\hat{H}_{\mathcal{SE}_{0}},\rho_{\mathcal{SE}_{0}}\left(t\right)\right]-k\rho_{\mathcal{SE}_{0}}\left(t\right).
\end{equation}
Then, each element $\ket{\phi_i}\bra{\phi_i}$ of the joint state evolves as 
\begin{align}
 & \dfrac{d}{dt}\braket{\phi_i|\rho_{\mathcal{SE}_{0}}\left(t\right)|\phi_j}\nonumber \\
 & =\bra{\phi_i}\left(-i\left[\hat{H}_{\mathcal{SE}_{0}},\rho_{\mathcal{SE}_{0}}\left(t\right)\right]-k\rho_{\mathcal{SE}\left(0\right)}\left(t\right)\right)\ket{\phi_j}\\
 & =\left(-i\Delta\omega_{i,j}-k\right)\braket{\phi_i|\rho_{\mathcal{SE}_{0}}\left(t\right)|\phi_j},
\end{align}
where $\Delta\omega_{i,j}=\omega_{i}-\omega_{j}$. Integrating over time $t\in\left[0,\tau_{se}\right)$ gives
\begin{align}
\braket{\phi_i|\rho_{\mathcal{SE}_{0}}\left(t\right)|\phi_j} & =e^{-kt}e^{-i\omega_{i}t}e^{i\omega_{j}t}\braket{\phi_i|\rho_{\mathcal{SE}_{0}}\left(0\right)|\phi_j}.
\end{align}
By using the unital operator $\mathbb{U}_{\tau}=\sum_{i}e^{-i\omega_i \tau}\ket{\phi_i}\bra{\phi_i}$, we can write this as: $\rho_{\mathcal{SE}_{0}}\left(t\right) = e^{-kt}\mathbb{U}_{t}\rho_{\mathcal{SE}_{0}}\left(0\right)\mathbb{U}_{t}^{\dagger}$.

As similar process occurs for times $\left[\tau_{se},\tau_{se}+\tau_{ee}\right)$
and so forth.

This leads to the recursive solution: starting from the initial state $\rho_{\mathcal{SE}_{0}}\left(0\right)=\rho_{\mathcal{S}}\left(0\right)\otimes\varrho_{E}$, for times $t\in\left[nT,nT+\tau_{se}\right)$ and $t\in\left[nT+\tau_{se},\left(n+1\right)T\right)$ respectively:
\begin{align}
\rho_{\mathcal{SE}_{n}}\left(t\right) & =e^{-k\left(t-nT\right)}\mathbb{U}_{t-nT}\rho_{\mathcal{SE}_{n}}\left(nT\right)\mathbb{U}_{t-nT}^{\dagger},\\
\rho_{\mathcal{SE}_{n}}\left(t\right) & =e^{-k\left(t-\left(nT+\tau_{se}\right)\right)}\mathbb{W}_{t-\left(nT+\tau_{se}\right)}\nonumber \\
 & \qquad\times\rho_{\mathcal{SE}_{n}}\left(nT+\tau_{se}\right)\mathbb{W}_{t-\left(nT+\tau_{se}\right)}^{\dagger}.
\end{align}

\subsection{Proof of Lemma \ref{lem:RPM_unital_coherence}}

\begin{proof}
To see this, we consider the converse, i.e., if $\rho_{\mathcal{S}}\left(0\right)=\left(\id_{2}/2\right)^{\otimes3}$ and where $\rho_{\mathcal{E}_{n,}}\left(0\right)=\left(\id_{2}/2\right)^{\otimes3}$ (as all subsystems are two-level). Then, the state evolves as:
\begin{align}
\rho_{\mathcal{SE}_{0}}\left(t\right) & =e^{-kt}\mathbb{U}_{t}\rho_{\mathcal{SE}_{0}}\left(0\right)\mathbb{U}_{t}^{\dagger}\\
 & =e^{-kt}\left(\dfrac{\id_{2}}{2}\right)^{\otimes 6},\quad t\in\left[0,\tau_{se}\right)\\
\rho_{\mathcal{SE}_{0}}\left(t\right) & =e^{-k\left(t-\tau_{se}\right)}\mathbb{W}_{t-\tau_{se}}\rho_{\mathcal{SE}_{0}}\left(\tau_{se}\right)\mathbb{W}_{t-\tau_{se}}^{\dagger}\\
 & =e^{-k\left(t-\tau_{se}\right)}e^{-k\tau_{se}}\left(\dfrac{\id_{2}}{2}\right)^{\otimes 6},\quad t\in\left[\tau_{se},T\right),
\end{align}
and so forth, where $\mathbb{U}$ and $\mathbb{W}$ are the unitary operator components of the system-environment evolution given in the \hyperref[sec:methods]{Methods} subsection above for general Hamiltonians. The system state is always:
\begin{equation}
\rho_{\mathcal{S}}\left(t\right)=e^{-kt}\left(\dfrac{\id_{2}}{2}\right)^{\otimes 3},
\end{equation}
for all field angles $\left(\theta,\phi\right)$. Hence, the singlet yield is the same for all field angles, leading to zero performance $\Delta [\varphi_\text{singlet}] = 0$.
\end{proof}


\subsection{Proof of Lemma \ref{lemma:necessary}}

\begin{proof}
Recall that the complete system-environment evolution is given in Eq.~(\ref{eq:evolution}), with unitary component given by $\hat{H}_{total} = \hat{H}_0+\hat{V}_{\mathcal{SE}}+\hat{H}_{\vec{\boldsymbol{B}}}(\theta,\phi)$. 

If either $P_{\text{singlet}}^{\left(A,B\right)}$ or $\rho_{\mathcal{SE}}\left(0\right)$ commute with the unitary evolution, then they become fixed points of the evolution \citep{Arias2002}, in which case $\tr[P_{\text{singlet}}^{\left(A,B\right)}\rho_{\mathcal{SE}}\left(t\right)]=e^{-kt}\tr[P_{\text{singlet}}^{\left(A,B\right)}\rho_{\mathcal{SE}}\left(0\right)]$ is the same for all angles [where the $e^{-kt}$ decay factor comes from the recombination in Eq.~(\ref{eq:evolution})]. Hence non-commutativity for both are necessary, i.e. $[P_{\text{singlet}}^{\left(A,B\right)},\hat{H}_{total}]\neq0$ and $[\rho_{\mathcal{SE}}\left(0\right),\hat{H}_{total}]\neq0$. Since $[P_{\text{singlet}}^{\left(A,B\right)},\hat{H}_{\vec{\boldsymbol{B}}}]=0$ commutes for all angles (due to the SU(2) symmetry in $P_{\text{singlet}}^{\left(A,B\right)}$ \citep{Keens2018}), this reduces to the form Eq.~(\ref{eq:necessary_non_trivial_P_H_V}) given in the lemma.

Finally, if the Zeeman interaction $\hat{H}_{\vec{\boldsymbol{B}}} (\theta,\phi)$ commutes with all the other interactions, $\hat{H}_\text{other} = \hat{H}_{0}+\hat{V}_{\mathcal{SE}}$ (which would mean that the other interactions have local SU(2) symmetry on the system), then $e^{-it(\hat{H}_{\vec{\boldsymbol{B}}} + \hat{H}_\text{other})}=e^{-it\hat{H}_{\vec{\boldsymbol{B}}}} e^{-it \hat{H}_\text{other}}$, and the effect of Zeeman interaction is unseen due to the singlet's symmetry:
\begin{align}
& \tr [P_{\text{singlet}}^{\left(A,B\right)} e^{-it(\hat{H}_{\vec{\boldsymbol{B}}} + \hat{H}_\text{other})} \rho_{\mathcal{SE}}\left(0\right) e^{it(\hat{H}_{\vec{\boldsymbol{B}}} + \hat{H}_\text{other})}] \nonumber \\
& = \tr \left[e^{it\hat{H}_{\vec{\boldsymbol{B}}}} P_{\text{singlet}}^{\left(A,B\right)} e^{-it\hat{H}_{\vec{\boldsymbol{B}}}} e^{-it \hat{H}_\text{other}} \rho_{\mathcal{SE}}\left(0\right) e^{it \hat{H}_\text{other}}\right] \\
& =  \tr \left[P_{\text{singlet}}^{\left(A,B\right)} e^{-it \hat{H}_\text{other}} \rho_{\mathcal{SE}}\left(0\right) e^{it \hat{H}_\text{other}}\right]
\end{align}
and subsequently the recombination yield would have no dependence on field angle.
\end{proof}

\subsection{Results in support for the necessity and sufficiency of basis-independent system-environment coherence and non-trivial system-environment interactions \label{methods:theorem_statements} }

In this subsection, we state the propositions and lemmas that support Proposition \ref{prop:necessary_conditions_hold}.

\begin{proposition}
For the model we consider (Sec. \ref{sec:Quantum-magnetic-sensor} with parameters in Tab.~\ref{tab:Parameters-and-configurations}), $[\rho_{\mathcal{S}}\otimes\id_{\mathcal{E}}/d_{\mathcal{E}},\hat{H}_{ex}+\hat{V}_{\mathcal{SE}}+\hat{H}_{\vec{\boldsymbol{B}}}\left(\theta,\phi\right)]=0$ if and only if $\rho_{\mathcal{S}}=\id_{\mathcal{S}}/d_{\mathcal{S}}$. 
\label{prop:rho_S_H_our_model}
\end{proposition}

The proof is given in Appendix \ref{sec:proof_of_rho_S_H_our_model}. Briefly, it proceeds by writing $\rho_{\mathcal{S}} = \sum_{ijk=0,x,y,z} p_{ijk} \sigma_i^A \otimes \sigma_j^B \otimes \sigma_k^C$ and looks at the coefficients of independent terms $\sigma_i^A \otimes \sigma_j^B \otimes \sigma_k^B \otimes \sigma_l^{\mathcal{E}_A} \otimes \sigma_m^{\mathcal{E}_B} \otimes \sigma_n^{\mathcal{E}_C}$. Our particular system-environment interaction mimics decoherence on the system and imposes the $\{\id,\sigma_z\}$ basis on the system and the form of $\hat{H}_{ex}$ and $\hat{H}_{\vec{\boldsymbol{B}}}$ reduce the required system down to the maximally mixed state.

\begin{lem}
If $\hat{V}_{\mathcal{SE}}=\hat{V}_{A\mathcal{E}_{A}}+\hat{V}_{B\mathcal{E}_{B}}+\hat{V}_{C\mathcal{E}_{C}}\neq0$ has nontrivial interaction between system $\mathcal{S}$ and environment $\mathcal{E}$ (i.e. contains two-body interaction terms), then $[P^{(A,B)}_{\text{singlet}},\hat{V}_{\mathcal{SE}}]\neq0$.\label{lemma:necessary-VSE}
\end{lem}
The proof of Lemma \ref{lemma:necessary-VSE} is given in Appendix \ref{methods:LemmaVSE}. Briefly, we write $\hat{V}_{A\mathcal{E}_{A}}=\sum_{ij=0,\ldots,3}g_{ij}\sigma_{i}^{A}\otimes\sigma_{j}^{\mathcal{E}_{A}}$, for $\sigma_{i}\in\left\{ \id,\sigma_{x},\sigma_{y},\sigma_{z}\right\} $ (and so forth for $\hat{V}_{B\mathcal{E}_{B}}$). Due to how each $\hat{V}_{A\mathcal{E}_{A}}$, $\hat{V}_{B\mathcal{E}_{B}}$ and $P^{(A,B)}_\text{singlet}$ all act on slightly different combinations of systems and environment, their commutator is only zero when $\hat{V}_{A\mathcal{E}_{A}}$ act \emph{locally} on the system $A$ and locally on environment $\mathcal{E}_A$, which contradicts the fact that it is an interaction between system and environment.

\begin{lem}
If $\hat{V}_{\mathcal{SE}}\neq0$ has nontrivial interaction between system and environment, then $[P^{(A,B)}_{\text{singlet}},\hat{H}_0]\neq-[P^{(A,B)}_{\text{singlet}},\hat{V}_{\mathcal{SE}}]$ for any system-Hamiltonian $\hat{H}_0$. \label{lem:P_AB_H_VSE}
\end{lem}
\begin{proof}
From the previous Lemma \ref{lemma:necessary-VSE}, we see that $[P^{(A,B)}_{\text{singlet}},\hat{V}_{\mathcal{SE}}]$ has at least one nontrivial term on the environment $\sigma_{j}^{\mathcal{E}_A}$, $j\in\{x,y,z\}$ that is nonzero. In contrast, $[P^{(A,B)}_{\text{singlet}},\hat{H}_0]=[P^{(A,B)}_{\text{singlet}},\hat{H}_0]\otimes \id_{\mathcal{E}}$. Since $\id_{\mathcal{E}}$ and $\sigma_{j=x,y,z}^\mathcal{E}$ are linearly independent, the two commutators contain independent terms and cannot be equal.
\end{proof}

This shows that the necessary condition $\left[P^{(A,B)}_{\text{singlet}},\hat{H}_0+\hat{V}_{\mathcal{SE}}\right]\neq0$ from Lemma \ref{lemma:necessary} is \emph{always} true provided the
nontrivial interaction $\hat{V}_{\mathcal{SE}}$.

\begin{lem}
$[\hat{H}_{\vec{\boldsymbol{B}}} (\theta,\phi),\hat{H}_0+\hat{V}_{\mathcal{SE}}]\neq 0$ for all nontrivial $\hat{V}_{\mathcal{SE}}$.
\label{lem:H_B_H_0_V}
\end{lem}

The proof is given in Appendix~\ref{sec:proof_of_lem_HB_H0_V}. It proceeds by contradiction and shows that the if the commutator is zero, then  $\hat{V}_{\mathcal{SE}}$ must be trivial (i.e. consisting of only local terms on the systems  $\mathcal{S}$ and environments $\mathcal{E}$). The proof also assumes that the magnetic field $\hat{H}_{\vec{\boldsymbol{B}}}$ does not act on the environment.

\section*{Acknowledgements}

This work was supported by the Engineering and Physical Sciences Research Council {[}grant number EP/L015242/1{]} and by the Gordon and Betty Moore Foundation [Grant number GBMF8820].

\section*{Author Contributions}

AOC conceived the presented idea. TPL performed the numerical computations and calculations. All authors discussed the results and contributed to the final manuscript.

\phantomsection
\addcontentsline{toc}{section}{References}

\bibliographystyle{apsrev4-1}
\bibliography{biblio}


\appendix

\begin{widetext}

\section{Various proofs}

\subsection{Proof of Lemma \ref{lemma:reduced_rho_S_HB_only}\label{sec:Lemmarho_HB}}

Recall the statement of Lemma \ref{lemma:reduced_rho_S_HB_only}: $[\rho_{\mathcal{S}}\left(0\right),\hat{H}_{\vec{\boldsymbol{B}}}(\theta,\phi)]=0$ for all angles $(\theta,\phi)$ if and only if
\begin{align}
\rho_{\mathcal{S}} =&\dfrac{\id_{ABC}}{8}+p_{AB}\vec{\sigma}^{A}\cdot\vec{\sigma}^{B}+p_{AC}\vec{\sigma}^{A}\cdot\vec{\sigma}^{C}+p_{BC}\vec{\sigma}^{B}\cdot\vec{\sigma}^{C} \nonumber \\
&+p_{ABC}\left(\sigma_{x}\otimes\sigma_{y}\otimes\sigma_{z}+\sigma_{z}\otimes\sigma_{x}\otimes\sigma_{y}+\sigma_{y}\otimes\sigma_{z}\otimes\sigma_{x}\right) \nonumber\\
&-p_{ABC}\left(\sigma_{x}\otimes\sigma_{z}\otimes\sigma_{y}+\sigma_{y}\otimes\sigma_{x}\otimes\sigma_{z}+\sigma_{z}\otimes\sigma_{y}\otimes\sigma_{x}\right).
\end{align}

The state above satisfies the zero commutation relation---this can be seen by straightforward calculation.

In the reverse direction: we write the system state in the following general form $\rho_{\mathcal{S}}=\sum_{ijk=0,x,y,z}p_{ijk}\sigma_{i}\otimes\sigma_{j}\otimes\sigma_{k}$, where $\sigma_0 = \id$. The commutator expands as follows:
\begin{align}
\left[\rho_{\mathcal{S}}\left(0\right),\hat{H}_{\vec{\boldsymbol{B}}}(\theta,\phi)\right]= & \dfrac{\gamma}{2}\sum_{ijk=0,x,y,z}p_{ijk}\left[\sigma_{i}\otimes\sigma_{j}\otimes\sigma_{k},\vec{B}\cdot\vec{\sigma}^{A}+\vec{B}\cdot\vec{\sigma}^{B}+\vec{B}\cdot\vec{\sigma}^{C}\right]\overset{!}{=}0.
\end{align}
By writing $\vec{B}\cdot\vec{\sigma}=B_{x}\sigma_{x}+B_{y}\sigma_{y}+B_{z}\sigma_{z}$, we further expand collect/separate the linearly independent terms, i.e. in the matrix basis $\left\{ \id,\sigma_{x},\sigma_{y},\sigma_{z}\right\} ^{\otimes3}$,
\begin{align}
& \left[\rho_{\mathcal{S}}\left(0\right),  \hat{H}_{\vec{\boldsymbol{B}}}(\theta,\phi)\right] \\
&  =  i\gamma \left(\sigma_{x}\left(p_{y00}B_{z}-p_{z00}B_{y}\right)+\sigma_{y}\left(p_{z00}B_{x}-p_{x00}B_{z}\right)+\sigma_{z}\left(p_{x00}B_{y}-p_{y00}B_{x}\right)\right)\otimes\id\otimes\id\label{eq:rho_HB_1}\\
 & +i\gamma\id\otimes\left(\sigma_{x}\left(p_{0y0}B_{z}-p_{0z0}B_{y}\right)+\sigma_{y}\left(p_{0z0}B_{x}-p_{0x0}B_{z}\right)+\sigma_{z}\left(p_{0x0}B_{y}-p_{0y0}B_{x}\right)\right)\otimes\id\label{eq:rho_HB_2}\\
 & +i\gamma\id\otimes\id\otimes\left(\sigma_{x}\left(p_{00y}B_{z}-p_{00z}B_{y}\right)+\sigma_{y}\left(p_{00z}B_{x}-p_{00x}B_{z}\right)+\sigma_{z}\left(p_{00x}B_{y}-p_{00y}B_{x}\right)\right)\label{eq:rho_HB_3}\\
 & +i\gamma\id\otimes\left[\begin{array}{l}
\sum_{k=x,y,z}\left(\sigma_{x}\left(p_{0yk}B_{z}-p_{0zk}B_{y}\right)+\sigma_{y}\left(p_{0zk}B_{x}-p_{0xk}B_{z}\right)+\sigma_{z}\left(p_{0xk}B_{y}-p_{0yk}B_{x}\right)\right)\otimes\sigma_{k}\\
+\sum_{j=x,y,z}\sigma_{j}\otimes\left(\sigma_{x}\left(p_{0jy}B_{z}-p_{0jz}B_{y}\right)+\sigma_{y}\left(p_{0jz}B_{x}-p_{0jx}B_{z}\right)+\sigma_{z}\left(p_{0jx}B_{y}-p_{0jy}B_{x}\right)\right)
\end{array}\right]\label{eq:rho_HB_4}\\
 & +i\gamma\left[\begin{array}{l}
\sum_{k=x,y,z}\left(\sigma_{x}\left(p_{y0k}B_{z}-p_{z0k}B_{y}\right)+\sigma_{y}\left(p_{z0k}B_{x}-p_{x0k}B_{z}\right)+\sigma_{z}\left(p_{x0k}B_{y}-p_{y0k}B_{x}\right)\right)\otimes\id\otimes\sigma_{k}\\
+\sum_{i=x,y,z}\sigma_{i}\otimes\id\otimes\left(\sigma_{x}\left(p_{i0y}B_{z}-p_{i0z}B_{y}\right)+\sigma_{y}\left(p_{i0z}B_{x}-p_{i0x}B_{z}\right)+\sigma_{z}\left(p_{i0x}B_{y}-p_{i0y}B_{x}\right)\right)
\end{array}\right]\label{eq:rho_HB_5}\\
 & +i\gamma\left[\begin{array}{l}
\sum_{j=x,y,z}\left(\sigma_{x}\left(p_{yj0}B_{z}-p_{zj0}B_{y}\right)+\sigma_{y}\left(p_{zj0}B_{x}-p_{xj0}B_{z}\right)+\sigma_{z}\left(p_{xj0}B_{y}-p_{yj0}B_{x}\right)\right)\otimes\sigma_{j}\\
+\sum_{i=,x,y,z}\sigma_{i}\otimes\left(\sigma_{x}\left(p_{iy0}B_{z}-p_{iz0}B_{y}\right)+\sigma_{y}\left(p_{iz0}B_{x}-p_{ix0}B_{z}\right)+\sigma_{z}\left(p_{ix0}B_{y}-p_{iy0}B_{x}\right)\right)
\end{array}\right]\otimes\id\label{eq:rho_HB_6}\\
 & +i\gamma\left\{ \begin{array}{l}
\sum_{jk=x,y,z}\left(\sigma_{x}\left(p_{yjk}B_{z}-p_{zjk}B_{y}\right)+\sigma_{y}\left(p_{zjk}B_{x}-p_{xjk}B_{z}\right)+\sigma_{z}\left(p_{xjk}B_{y}-p_{yjk}B_{x}\right)\right)\otimes\sigma_{j}\otimes\sigma_{k}\\
+\sum_{ik=x,y,z}\sigma_{i}\otimes\left(\sigma_{x}\left(p_{iyk}B_{z}-p_{izk}B_{y}\right)+\sigma_{y}\left(p_{izk}B_{x}-p_{ixk}B_{z}\right)+\sigma_{z}\left(p_{ixk}B_{y}-p_{iyk}B_{x}\right)\right)\otimes\sigma_{k}\\
+\sum_{ij=x,y,z}\sigma_{i}\otimes\sigma_{j}\otimes\left(\sigma_{x}\left(p_{ijy}B_{z}-p_{ijz}B_{y}\right)+\sigma_{y}\left(p_{ijz}B_{x}-p_{ijx}B_{z}\right)+\sigma_{z}\left(p_{ijx}B_{y}-p_{ijy}B_{x}\right)\right)
\end{array}\right\} \label{eq:rho_HB_7}\\
 &  \overset{!}{=}0.
\end{align}
Each of the lines above are linearly independent from each other. Thus, we will consider each separately.

For Eq. (\ref{eq:rho_HB_1}):
\begin{align}
 & i\gamma\left(\sigma_{x}\left(p_{y00}B_{z}-p_{z00}B_{y}\right)+\sigma_{y}\left(p_{z00}B_{x}-p_{x00}B_{z}\right)+\sigma_{z}\left(p_{x00}B_{y}-p_{y00}B_{x}\right)\right)\otimes\id\otimes\id\overset{!}{=}0\nonumber \\
\implies & p_{y00}B_{z}\overset{!}{=}p_{z00}B_{y},\qquad p_{z00}B_{x}\overset{!}{=}p_{x00}B_{z},\qquad p_{x00}B_{y}\overset{!}{=}p_{y00}B_{x}.
\end{align}
We wish to find the parameters $p_{ijk}$ in which these equations hold for all allowed $\left(B_{x},B_{y},B_{z}\right)$ magnetic fields. If,
for example the magnetic field is in the $z$ direction, meaning that $B_z\neq0$ and $B_x = B_y = 0$, then  $p_{y00}B_{z}=p_{z00}B_{y}$ implies that $p_{y00}=0$ etc. By considering different directions, the above holds for all
field angles if and only if
\begin{equation}
p_{x00}=p_{y00}=p_{z00}=0.
\end{equation}
Similarly, Eqs. (\ref{eq:rho_HB_2}) and (\ref{eq:rho_HB_3}) give:
\begin{align}
p_{0x0} & =p_{0y0}=p_{0z0}=0\\
p_{00x} & =p_{00y}=p_{00z}=0.
\end{align}

Now consider Eq. (\ref{eq:rho_HB_4}), which expands out to:
\begin{align}
0\overset{!}{=} & \id\otimes\sigma_{x}\otimes\sigma_{x}\left(\left(p_{0yx}B_{z}-p_{0zx}B_{y}+p_{0xy}B_{z}-p_{0xz}B_{y}\right)\right)\\
&+\id\otimes\sigma_{x}\otimes\sigma_{y}\left(p_{0yy}B_{z}-p_{0zy}B_{y}+p_{0xz}B_{x}-p_{0xx}B_{z}\right) \\
& +\id\otimes\sigma_{x}\otimes\sigma_{z}\left(p_{0yz}B_{z}-p_{0zz}B_{y}+p_{0xx}B_{y}-p_{0xy}B_{x}\right)\\
& +\id\otimes\sigma_{y}\otimes\sigma_{x}\left(p_{0zx}B_{x}-p_{0xx}B_{z}+p_{0yy}B_{z}-p_{0yz}B_{y}\right)\\
&+\id\otimes\sigma_{y}\otimes\sigma_{y}\left(p_{0zy}B_{x}-p_{0xy}B_{z}+p_{0yz}B_{x}-p_{0yx}B_{z}\right)\\
&+\id\otimes\sigma_{y}\otimes\sigma_{z}\left(p_{0zz}B_{x}-p_{0xz}B_{z}+p_{0yx}B_{y}-p_{0yy}B_{x}\right)\\
& +\id\otimes\sigma_{z}\otimes\sigma_{x}\left(p_{0xx}B_{y}-p_{0yx}B_{x}+p_{0zy}B_{z}-p_{0zz}B_{y}\right)\\
&+\id\otimes\sigma_{z}\otimes\sigma_{y}\left(p_{0xy}B_{y}-p_{0yy}B_{x}+p_{0zz}B_{x}-p_{0zx}B_{z}\right)\\
& +\id\otimes\sigma_{z}\otimes\sigma_{z}\left(p_{0xz}B_{y}-p_{0yz}B_{x}+p_{0zx}B_{y}-p_{0zy}B_{x}\right),
\end{align}
and each of those coefficients must be zero. For the coefficient of $\id\otimes\sigma_{x}\otimes\sigma_{x}$,
\begin{align}
p_{0yx}B_{z}-p_{0zx}B_{y}+p_{0xy}B_{z}-p_{0xz}B_{y} & \overset{!}{=}0\\
\implies B_{z}\left(p_{0yx}+p_{0xy}\right) & =B_{y}\left(p_{0zx}+p_{0xz}\right),
\end{align}
 which holds for all magnetic field directions iff
\begin{equation}
p_{0yx}=-p_{0xy},\qquad p_{0zx}=-p_{0xz},
\end{equation}
(i.e. by considering the value of the parameters when the field is \emph{only} in $x,y$, or $z$ direction). For the coefficient of $\id\otimes\sigma_{x}\otimes\sigma_{y}$,
\begin{align}
p_{0yy}B_{z}-p_{0zy}B_{y}+p_{0xz}B_{x}-p_{0xx}B_{z} & \overset{!}{=}0\\
\implies\left(p_{0yy}-p_{0xx}\right)B_{z} & =p_{0zy}B_{y}-p_{0xz}B_{x},\\
\implies & p_{0yy}=p_{0xx},\quad p_{0xz}=0,\quad p_{0zy}=0.
\end{align}
By considering all the linearly independent coefficients of  Eq. (\ref{eq:rho_HB_4}), we obtain in summary:
\begin{align}
p_{0xx} & =p_{0yy}=p_{0zz}\\
p_{0jk} & =0\quad\text{for }j\neq k,\quad j,k\in\left\{ x,y,z\right\} .
\end{align}
Similarly (by considering the permutations), Eqs. (\ref{eq:rho_HB_5}) and (\ref{eq:rho_HB_6}) give us:
\begin{align}
p_{x0x} & =p_{y0y}=p_{z0z}\\
p_{i0k} & =0\quad\text{for }i\neq k,\quad i,k\in\left\{ x,y,z\right\} \\
p_{xx0} & =p_{yy0}=p_{zz0}\\
p_{ij0} & =0\quad\text{for }i\neq j,\quad i,j\in\left\{ x,y,z\right\} .
\end{align}
At this point, our current reduced system state is 
\begin{equation}
\rho_{\mathcal{S}}=p_{000}\id\otimes\id\otimes\id+p_{xx0}\vec{\sigma}^{A}\cdot\vec{\sigma}^{B}\otimes\id^{C}+p_{x0x}\vec{\sigma}^{A}\cdot\vec{\sigma}^{C}\otimes\id^{B}+p_{0xx}\id^{A}\otimes\vec{\sigma}^{B}\cdot\vec{\sigma}^{C}+\sum_{i,j,k=,x,y,z}p_{ijk}\sigma_{i}\otimes\sigma_{j}\otimes\sigma_{k}.
\end{equation}

Finally, for Eq. (\ref{eq:rho_HB_7}):
\begin{align}
0\overset{!}{=} & \sum_{ijk=x,y,z}\sigma_{i}\left(p_{i\oplus1,jk}B_{i\oplus2}-p_{i\oplus2,jk}B_{i\oplus1}\right)\otimes\sigma_{j}\otimes\sigma_{k} +\sum_{ijk=x,y,z}\sigma_{i}\otimes\left(\sigma_{j}\left(p_{i,j\oplus1,k}B_{j\oplus2}-p_{i,j\oplus2,k}B_{j\oplus1}\right)\right)\otimes\sigma_{k}\nonumber \\
 & +\sum_{ijk=x,y,z}\sigma_{i}\otimes\sigma_{j}\otimes\sigma_{k}\left(p_{i,j,k\oplus1}B_{k\oplus2}-p_{i,j,k\oplus2}B_{k\oplus1}\right)\nonumber \\
= & \sum_{ijk=x,y,z}\sigma_{i}\otimes\sigma_{j}\otimes\sigma_{k}\left\{  \begin{array}{r}
p_{i\oplus1,jk}B_{i\oplus2}-p_{i\oplus2,jk}B_{i\oplus1}+p_{i,j\oplus1,k}B_{j\oplus2}-p_{i,j\oplus2,k}B_{j\oplus1}\\
+p_{i,j,k\oplus1}B_{k\oplus2}-p_{i,j,k\oplus2}B_{k\oplus1}
\end{array}\right\} \\
\implies0\overset{!}{=} & 
p_{i\oplus1,jk}B_{i\oplus2}-p_{i\oplus2,jk}B_{i\oplus1}+p_{i,j\oplus1,k}B_{j\oplus2}-p_{i,j\oplus2,k}B_{j\oplus1} \nonumber\\
& \phantom{p_{i\oplus1,jk}B_{i\oplus2}-p_{i\oplus2,jk}B_{i\oplus1}} +p_{i,j,k\oplus1}B_{k\oplus2}-p_{i,j,k\oplus2}B_{k\oplus1}, \qquad\forall i,j,k\in\left\{ x,y,z\right\},
\end{align}
where the addition is modulo $3$, ordered $\left\{ x,y,z\right\} $ so that, for example, $z\oplus1=x$.

Now consider if $i=j$, $k=i\oplus1$:
\begin{align}
0\overset{!}{=} & B_{i\oplus1}\left(-p_{i\oplus2,i,i\oplus1}-p_{i,i\oplus2,i\oplus1}\right)+B_{i\oplus2}\left(p_{i\oplus1,i,i\oplus1}+p_{i,i\oplus1,i\oplus1}\right)+p_{i,i,i\oplus2}B_{i}-p_{i,i,i}B_{i\oplus2}\\
= & B_{i\oplus1}\left(-p_{i\oplus2,i,i\oplus1}-p_{i,i\oplus2,i\oplus1}\right)+B_{i\oplus2}\left(p_{i\oplus1,i,i\oplus1}+p_{i,i\oplus1,i\oplus1}-p_{i,i,i}\right)+p_{i,i,i\oplus2}B_{i}.
\end{align}
This must hold for all magnetic field directions, including if $B_{i}=B_{i\oplus1}=0$
or if $B_{i}=B_{i\oplus2}=0$ or if $B_{i\oplus1}=B_{i\oplus2}=0$.
This leads to the following conditions:
\begin{align}
p_{i\oplus1,i,i\oplus1}+p_{i,i\oplus1,i\oplus1}=p_{i,i,i},\qquad p_{i\oplus2,i,i\oplus1} & =-p_{i,i\oplus2,i\oplus1},\qquad p_{i,i,i\oplus2}=0.\label{eq:rho_HB_7_iii+1}
\end{align}
Similarly, with permutation, i.e. for $i=k$, $j=i\oplus1$ and for
$j=k$, $i=j\oplus1$, we obtain the following conditions:
\begin{eqnarray}
p_{i\oplus1,i\oplus1,i}+p_{i,i\oplus1,i\oplus1}=p_{i,i,i},\qquad & p_{i\oplus2,i\oplus1,i}=-p_{i,i\oplus1,i\oplus2}, & \quad p_{i,i\oplus2,i}=0,\label{eq:rho_HB_7_ii+1i}\\
p_{j\oplus1,j\oplus1,j}+p_{j\oplus1,j,j\oplus1}=p_{j,jj},\qquad & p_{j\oplus1,j\oplus2,j}=-p_{j\oplus1,j,j\oplus2}, & \quad p_{j\oplus2,jj}=0.\label{eq:rho_HB_7_i+1ii}
\end{eqnarray}

Since $p_{i,i,i\oplus1}=p_{i,i\oplus1,i}=p_{i\oplus1,i,}=p_{i,i,i\oplus2}=p_{i,i\oplus2,i}=p_{i\oplus2,i,i}=0$,
we also have that $p_{i\oplus1,i,i\oplus1}+p_{i,i\oplus1,i\oplus1}=p_{i,i,i}=0$.
Hence, for, $i,j,k\in\left\{ x,y,z\right\} $, the only nonzero $p_{ijk}$
terms are when all $i\neq j\neq k$. The second conditions of Eqs.
(\ref{eq:rho_HB_7_iii+1}), (\ref{eq:rho_HB_7_ii+1i}), and (\ref{eq:rho_HB_7_i+1ii})
give:
\begin{align}
 & p_{i\oplus2,i,i\oplus1}=-p_{i,i\oplus2,i\oplus1}\implies p_{i,i\oplus1,i\oplus2}=-p_{i\oplus1,i,i\oplus2}\\
 & p_{i\oplus2,i\oplus1,i}=-p_{i,i\oplus1,i\oplus2}\implies p_{i,i\oplus1,i\oplus2}=-p_{i\oplus2,i\oplus1,i}\implies p_{i\oplus2,i,i\oplus1}=-p_{i\oplus1,i,i\oplus2}\\
 & p_{j\oplus1,j\oplus2,j}=-p_{j\oplus1,j,j\oplus2}\implies p_{i,i\oplus1,i\oplus2}=-p_{i,i\oplus2,i\oplus1}\implies p_{i\oplus1,i\oplus2,i}=-p_{i\oplus1,i,i\oplus2},
\end{align}
so we have that:
\begin{equation}
p_{i,i\oplus1,i\oplus2}=p_{i\oplus2,i,i\oplus1}=p_{i\oplus1,i\oplus2,i}=-p_{i\oplus1,i,i\oplus2}=-p_{i\oplus2,i\oplus1,i}=-p_{i,i\oplus2,i\oplus1}.
\end{equation}
 This, in short, gives us that 
\begin{equation}
p_{xyz}=p_{yzx}=p_{zxy}=-p_{xzy}=-p_{yxz}=-p_{zyx}.
\end{equation}
Our reduced state is now
\begin{align}
\rho_{\mathcal{S}}= &p_{000}\id\otimes\id\otimes\id+p_{xx0}\vec{\sigma}^{A}\cdot\vec{\sigma}^{B}\otimes\id^{C}+p_{x0x}\vec{\sigma}^{A}\cdot\vec{\sigma}^{C}\otimes\id^{B}+p_{0xx}\id^{A}\otimes\vec{\sigma}^{B}\cdot\vec{\sigma}^{C} \nonumber \\
& +p_{xyz}\sum_{i=,x,y,z}\sigma_{i}\otimes\sigma_{i\oplus1}\otimes\sigma_{i\oplus2}-p_{xyz}\sum_{i=,x,y,z}\sigma_{i}\otimes\sigma_{i\oplus2}\otimes\sigma_{i\oplus1},
\end{align}
and with straightforward calculation, satisfies $[\rho_{\mathcal{S}},\hat{H}_{\vec{\boldsymbol{B}}}\left(\theta,\phi\right)]=0$ for all $\left(\theta,\phi\right)$. Finally, state normalisation gives $p_{000}=\id_{8}/8$. \hfill $\square$

\subsection{Proof of Proposition \ref{prop:rho_S_H_our_model}\label{sec:proof_of_rho_S_H_our_model}}

Recall that we want to show that  $[\rho_{\mathcal{S}}\otimes\dfrac{\id_{\mathcal{E}}}{d_{\mathcal{E}}},\hat{H}_{ex}+\hat{H}_{\vec{\boldsymbol{B}}}\left(\theta,\phi\right)+\hat{V}_{\mathcal{SE}}]=0$ if and only if $\rho_{\mathcal{S}}=\id_{\mathcal{S}}/d_{\mathcal{S}}$, when the initial environment state $\rho_{\mathcal{E}}=\id_{\mathcal{E}}/d_{\mathcal{E}}$, and with interactions
\begin{align}
\hat{V}_{\mathcal{SE}} & =J_{se}\sum_{\alpha=A,B,C}\ket{1}\bra{1}^{\alpha}\otimes\left(\id-\sigma_{x}\right)^{\mathcal{E}_{\alpha}},\\
\hat{H}_{ex} & =-\dfrac{J_{ABC}}{2}\sum_{\alpha<\beta}\left(\id+\boldsymbol{\hat{\sigma}}^\alpha\cdot\boldsymbol{\hat{\sigma}}^\beta\right).
\end{align}

\begin{proof}
It is immediate that $\rho_{\mathcal{S}}=\id_{\mathcal{S}}/d_{\mathcal{S}}$ satisfies the zero commutator.

To prove the other direction, note that we are taking our model in particular with system-environment interaction that we can write as
$\hat{V}_{\mathcal{SE}} =\frac{J_{se}}{2}\sum_{\alpha=A,B,C}\left(\id-\sigma_{z}\right)^{\alpha}\otimes\left(\id-\sigma_{x}\right)^{\mathcal{E}_{\alpha}}.$
Now, 
\begin{equation}
\left[\rho_{\mathcal{S}}\otimes\dfrac{\id_{\mathcal{E}}}{d_{\mathcal{E}}},\hat{H}_{ex}+\hat{H}_{\vec{\boldsymbol{B}}}\left(\theta,\phi\right)+\hat{V}_{\mathcal{SE}}\right]=\left[\rho_{\mathcal{S}},\hat{H}_{ex}+\hat{H}_{\vec{\boldsymbol{B}}}\left(\theta,\phi\right)\right]\otimes\dfrac{\id_{\mathcal{E}}}{d_{\mathcal{E}}}+\left[\rho_{\mathcal{S}}\otimes\dfrac{\id_{\mathcal{E}}}{d_{\mathcal{E}}},\hat{V}_{\mathcal{SE}}\right]\overset{!}{=}0.
\end{equation}
The second commutator contains terms of form $\sigma_{x,y,z}^{\mathcal{E}}$ that are absent from the first commutator. Hence, the coefficients of $\sigma_{x,y,z}^{\mathcal{E}}$ in the second commutator must independently (due to linear independence) be zero for the statement to hold. Thus, let us first consider the second commutator and determine the reduced
form of $\rho_{\mathcal{S}}$ needed. Let us write the general system state as $\rho_{\mathcal{S}} =\sum_{ijk=0,x,y,z}p_{ijk}\sigma_{i}^{A}\otimes\sigma_{j}^{B}\otimes\sigma_{k}^{C}$. 

Consider the interaction between system $A$ and its local bath: $\hat{V}_{A\mathcal{E}_{A}}=\dfrac{J_{se}}{2}\left(\id-\sigma_{z}\right)^{A}\otimes\left(\id-\sigma_{x}\right)^{\mathcal{E}_{A}}$ (the results for $\hat{V}_{B\mathcal{E}_{B}}$ and $\hat{V}_{C\mathcal{E}_{C}}$ will be analogous):
\begin{align}
\left[\rho_{\mathcal{S}}\otimes\dfrac{\id_{\mathcal{E}}}{d_{\mathcal{E}}},\hat{V}_{A\mathcal{E}_{A}}\right] & =\dfrac{J_{se}}{2}\sum_{ijk=0,x,y,z}p_{ijk}\left[\sigma_{i}^{A}\otimes\sigma_{j}^{B}\otimes\sigma_{k}^{C}\otimes\dfrac{\id_{\mathcal{E}}}{d_{\mathcal{E}}},\left(\id-\sigma_{z}\right)^{A}\otimes\left(\id-\sigma_{x}\right)^{\mathcal{E}_{A}}\right]\\
& =\dfrac{J_{se}}{2}\sum_{ijk=0,x,y,z}p_{ijk}\left[\sigma_{i}^{A}\otimes\dfrac{\id^{\mathcal{E}_A}}{2},\left(\id-\sigma_{z}\right)^{A}\otimes\left(\id-\sigma_{x}\right)^{\mathcal{E}_A}\right]\otimes\sigma_{j}^{B}\otimes\sigma_{k}^{C}\otimes\dfrac{\id^{\mathcal{E}_B\mathcal{E}_C}}{d_{\mathcal{E}_B\mathcal{E}_C}}\\
& =\dfrac{J_{se}}{2}\sum_{ijk=0,x,y,z}p_{ijk}\left[\sigma_{z}^{A},\sigma_{i}^{A}\right]\otimes\sigma_{j}^{B}\otimes\sigma_{k}^{C}\otimes\left(\dfrac{\id-\sigma_{x}}{2}\right)^{\mathcal{E}_A}\otimes\dfrac{\id^{\mathcal{E}_B\mathcal{E}_C}}{d_{\mathcal{E}_B\mathcal{E}_C}}.
\end{align}
The (only) term containing $\sigma_{x}^{\mathcal{E}_A}$ is the following, and it must be zero:
\begin{align}
\dfrac{J_{se}}{2} & \sum_{ijk=0,x,y,z}p_{ijk}\left[\sigma_{z}^{A},\sigma_{i}^{A}\right]\otimes\sigma_{j}^{B}\otimes\sigma_{k}^{C}\otimes\dfrac{-\sigma_{x}^{\mathcal{E}_A}}{2}\otimes\dfrac{\id^{\mathcal{E}_B\mathcal{E}_C}}{d_{\mathcal{E}_B\mathcal{E}_C}} \nonumber \\
=& \dfrac{J_{se}}{2}  \sum_{jk=0,x,y,z}\left(p_{xjk}2i\sigma_{y}^{A}-p_{yjk}2i\sigma_{x}^{A}\right)\otimes\sigma_{j}^{B}\otimes\sigma_{k}^{C}\otimes\dfrac{-\sigma_{x}^{\mathcal{E}_A}}{2}\otimes\dfrac{\id^{\mathcal{E}_B\mathcal{E}_C}}{d_{\mathcal{E}_B\mathcal{E}_C}}\overset{!}{=}0\\
\implies & p_{xjk}=p_{yjk}=0, \qquad \forall j,k\in\{0,x,y,z\},
\end{align}
due to linear independence of the different terms. Similarly, by  examining $\hat{V}_{B\mathcal{E}_B}$ and $\hat{V}_{C\mathcal{E}_C}$, we find that $p_{ixk}=p_{iyk}=0$ $\forall i,k\in \{0,x,y,z\}$ and $p_{ijx}=p_{ijy}=0$ $\forall i,j\in \{0,x,y,z\}$. This means that the system state has only terms with $\id=\sigma_{0}$ and $\sigma_{z}$:
\begin{equation}
\rho_{S|z}\coloneqq\sum_{ijk=0,z}p_{ijk}\sigma_{i}^{A}\otimes\sigma_{j}^{B}\otimes\sigma_{k}^{C}.
\end{equation}

Now in fact, this $\rho_{S|z}$ form has $\left[\rho_{S|z}\otimes \id_{\mathcal{E}}/d_{\mathcal{E}},\hat{V}_{\mathcal{SE}}\right]=0$. Thus, we now need to consider
\begin{align}
\left[\rho_{S|z},\hat{H}_{ex}+\hat{H}_{\vec{\boldsymbol{B}}}\left(\theta,\phi\right)\right] & \overset{!}{=}0.
\end{align}
Now if $\left[\rho_{S|z},\hat{H}_{ex}+\hat{H}_{\vec{\boldsymbol{B}}}\left(\theta,\phi\right)\right]\overset{!}{=}0$ at all angles, then consider if $\hat{H}_{\vec{\boldsymbol{B}}}\left(0,0\right)=\dfrac{\gamma B_{0}}{2}\left(\sigma_{z}^{A}+\sigma_{z}^{B}+\sigma_{z}^{C}\right)$. In this case,
\begin{equation}
\left[\rho_{S|z},\hat{H}_{\vec{\boldsymbol{B}}}\left(0,0\right)\right]=\left[\sum_{ijk=0,z}p_{ijk}\sigma_{i}^{A}\otimes\sigma_{j}^{B}\otimes\sigma_{k}^{C},\dfrac{\gamma B_{0}}{2}\left(\sigma_{z}^{A}+\sigma_{z}^{B}+\sigma_{z}^{C}\right)\right]=0,
\end{equation}
since all terms commute with each other. Hence, $\left[\rho_{S|z},\hat{H}_{ex}+\hat{H}_{\vec{\boldsymbol{B}}}\left(0,0\right)\right]\overset{!}{=}0$ implies $\left[\rho_{S|z},\hat{H}_{ex}\right]\overset{!}{=}0$. Expanding,
\begin{align}
&\left[\rho_{S|z},\hat{H}_{ex}\right] \nonumber \\
& =-\dfrac{J_{ABC}}{2}\sum_{ijk=0,z}\left[p_{ijk}\sigma_{i}^{A}\otimes\sigma_{j}^{B}\otimes\sigma_{k}^{C},3\id+\vec{\sigma}^{A}\cdot\vec{\sigma}^{B}+\vec{\sigma}^{A}\cdot\vec{\sigma}^{C}+\vec{\sigma}^{B}\cdot\vec{\sigma}^{C}\right]\\
 & =-\dfrac{J_{ABC}}{2}\sum_{ijk=0,z}p_{ijk}\left(\left[\sigma_{i}^{A}\otimes\sigma_{j}^{B},\vec{\sigma}^{A}\cdot\vec{\sigma}^{B}\right]\otimes\sigma_{k}^{C}+\left[\sigma_{i}^{A}\otimes\sigma_{k}^{C},\vec{\sigma}^{A}\cdot\vec{\sigma}^{C}\right]\otimes\sigma_{j}^{B}+\sigma_{i}^{A}\otimes\left[\sigma_{j}^{B}\otimes\sigma_{k}^{C},\vec{\sigma}^{B}\cdot\vec{\sigma}^{C}\right]\right). \label{eq:rhoS_Hex_expand}
\end{align}
Consider the first terms in Eq.~(\ref{eq:rhoS_Hex_expand}):
\begin{align}
& \sum_{ijk=0,z}p_{ijk}\left[\sigma_{i}^{A}\otimes\sigma_{j}^{B},\vec{\sigma}^{A}\cdot\vec{\sigma}^{B}\right]\otimes\sigma_{k}^{C}\nonumber \\
& =\sum_{ijk=0,z}p_{ijk}\left(\left[\sigma_{i}^{A}\otimes\sigma_{j}^{B},\sigma_{x}\otimes\sigma_{x}\right]+\left[\sigma_{i}^{A}\otimes\sigma_{j}^{B},\sigma_{y}\otimes\sigma_{y}\right]+\cancel{\left[\sigma_{i}^{A}\otimes\sigma_{j}^{B},\sigma_{z}\otimes\sigma_{z}\right]}\right)\otimes\sigma_{k}^{C}\\
 & =\sum_{k=0,z}\left(\begin{array}{l}
   \left[p_{0zk}\id\otimes\sigma_{z}+p_{z0k}\sigma_{z}\otimes\id+p_{zzk}\sigma_{z}\otimes\sigma_{z},\sigma_{x}\otimes\sigma_{x}\right]\\
   +\left[p_{0zk}\id\otimes\sigma_{z}+p_{z0k}\sigma_{z}\otimes\id + p_{zzk}\sigma_z\otimes\sigma_z,\sigma_{y}\otimes\sigma_{y}\right]
 \end{array} \right)\otimes\sigma_{k}^{C}\\
 & =\sum_{k=0,z}\left(p_{0zk}\sigma_{x}\otimes2i\sigma_{y}+p_{z0k}2i\sigma_{y}\otimes\sigma_{x}+p_{0zk}\sigma_{y}\otimes\left(-2i\sigma_{x}\right)+p_{z0k}\left(-2i\sigma_{x}\right)\otimes\sigma_{y}\right)\otimes\sigma_{k}^{C}\\
 & =2i\sum_{k=0,z}\left(\left(p_{0zk}-p_{z0k}\right)\sigma_{x}^{A}\otimes\sigma_{y}^{B}+\left(p_{z0k}-p_{0zk}\right)\sigma_{y}^{A}\otimes\sigma_{x}^{B}\right)\otimes\sigma_{k}^{C}.
\end{align}
Similarly,
\begin{align}
\sum_{ijk=0,z}p_{ijk}\left[\sigma_{i}^{A}\otimes\sigma_{k}^{C},\vec{\sigma}^{A}\cdot\vec{\sigma}^{C}\right]\otimes\sigma_{j}^{B} & =2i\sum_{j=0,z}\left(\left(p_{0jz}-p_{zj0}\right)\sigma_{x}^{A}\otimes\sigma_{y}^{C}+\left(p_{zj0}-p_{0jz}\right)\sigma_{y}^{A}\otimes\sigma_{x}^{C}\right)\otimes\sigma_{j}^{B}\\
\sum_{ijk=0,z}p_{ijk}\sigma_{i}^{A}\otimes\left[\sigma_{j}^{B}\otimes\sigma_{k}^{C},\vec{\sigma}^{B}\cdot\vec{\sigma}^{C}\right] & =2i\sum_{i=0,z}\sigma_{i}^{A}\otimes\left(\left(p_{i0z}-p_{iz0}\right)\sigma_{x}^{B}\otimes\sigma_{y}^{C}+\left(p_{iz0}-p_{i0z}\right)\sigma_{y}^{B}\otimes\sigma_{x}^{C}\right).
\end{align}
Hence,
\begin{align}
\left[\rho_{S|z},\hat{H}_{ex}\right] & =-J_{ABC}i\sum_{k=0,z}\left(\left(p_{0zk}-p_{z0k}\right)\sigma_{x}^{A}\otimes\sigma_{y}^{B}+\left(p_{z0k}-p_{0zk}\right)\sigma_{y}^{A}\otimes\sigma_{x}^{B}\right)\otimes\sigma_{k}^{C}\nonumber \\
 & \qquad-J_{ABC}i\sum_{j=0,z}\left(\left(p_{0jz}-p_{zj0}\right)\sigma_{x}^{A}\otimes\sigma_{y}^{C}+\left(p_{zj0}-p_{0jz}\right)\sigma_{y}^{A}\otimes\sigma_{x}^{C}\right)\otimes\sigma_{j}^{B}\nonumber \\
 & \qquad-J_{ABC}i\sum_{i=0,z}\sigma_{i}^{A}\otimes\left(\left(p_{i0z}-p_{iz0}\right)\sigma_{x}^{B}\otimes\sigma_{y}^{C}+\left(p_{iz0}-p_{i0z}\right)\sigma_{y}^{B}\otimes\sigma_{x}^{C}\right)\overset{!}{=}0.
\end{align}
However, note that each of the terms are linearly independent, and so the coefficients must be zero, which implies that
\begin{align}
p_{0zk} & =p_{z0k}\quad\text{for }k=0,z,\\
p_{0jz} & =p_{zj0}\quad\text{for }j=0,z,\\
p_{i0z} & =p_{iz0}\quad\text{for }i=0,z.
\end{align}
Hence we can reduce the system state again:
\begin{align}
\rho_{S|z}^{\prime}= & p_{000}\id\otimes\id\otimes\id+p_{zzz}\sigma_{z}\otimes\sigma_{z}\otimes\sigma_{z}\nonumber \\
 & +p_{z00}\left(\id\otimes\id\otimes\sigma_{z}+\id\otimes\sigma_{z}\otimes\id+\sigma_{z}\otimes\id\otimes\id\right)\\
 & +p_{zz0}\left(\id\otimes\sigma_{z}\otimes\sigma_{z}+\sigma_{z}\otimes\id\otimes\sigma_{z}+\sigma_{z}\otimes\sigma_{z}\otimes\id\right).\nonumber 
\end{align}

Now going back to $\left[\rho_{S|z}^{\prime},\hat{H}_{ex}+\hat{H}_{\vec{\boldsymbol{B}}}\left(\theta,\phi\right)\right]\overset{!}{=}0$,
since $\left[\rho_{S|z}^{\prime},\hat{H}_{ex}\right]=0$, consider $\left[\rho_{S|z}^{\prime},\hat{H}_{\vec{\boldsymbol{B}}}\left(\theta,\phi\right)\right]$.
Consider if the magnetic field is pointed in the $x$ direction, such that  $\hat{H}_{\vec{\boldsymbol{B}}}\left(0,\dfrac{\pi}{2}\right)=\dfrac{\gamma B_{0}}{2}\left(\sigma_{x}^{A}+\sigma_{x}^{B}+\sigma_{x}^{C}\right)$.
In this case, the commutator expands out to:
\begin{align}
\left[\rho_{S|z}^{\prime},\hat{H}_{\vec{\boldsymbol{B}}}\left(0,\dfrac{\pi}{2}\right)\right]
 & =\dfrac{\gamma B_{0}}{2}2i\sigma_{y}^{A}\otimes\left(p_{z00}\id^{B}\otimes\id^{C}+p_{zz0}\left(\id\otimes\sigma_{z}+\sigma_{z}\otimes\id\right)+p_{zzz}\sigma_{z}\otimes\sigma_{z}\right) \nonumber\\
 &\phantom{=} +\dfrac{\gamma B_{0}}{2}2i\sigma_{y}^{B}\otimes\left(p_{z00}\id^{A}\otimes\id^{C}+p_{zz0}\left(\id\otimes\sigma_{z}+\sigma_{z}\otimes\id\right)+p_{zzz}\sigma_{z}\otimes\sigma_{z}\right) \nonumber\\
 &\phantom{=} +\dfrac{\gamma B_{0}}{2}2i\sigma_{y}^{C}\otimes\left(p_{z00}\id^{A}\otimes\id^{B}+p_{zz0}\left(\id\otimes\sigma_{z}+\sigma_{z}\otimes\id\right)+p_{zzz}\sigma_{z}\otimes\sigma_{z}\right) \overset{!}{=}0\\
\implies & p_{z00}=0,\qquad p_{zz0}=0,\qquad p_{zzz}=0,
\end{align}
again due to the linear independence of the terms and subsequent requirements for zero coefficients. Thus, the reduced system state is $\rho_{\mathcal{S}}=p_{000}\id\otimes\id\otimes\id=\dfrac{\id_{\mathcal{S}}}{d_{\mathcal{S}}}$ with the $d_{\mathcal{S}}$ term for normalisation.
\end{proof}

\subsection{Proof of Lemma \ref{lemma:necessary-VSE} \label{methods:LemmaVSE}}

Let $P^{(A,B)}_{\text{singlet}}=P^{AB}$ to condense notation.
Recall the statement of Lemma  \ref{lemma:necessary-VSE}: \emph{If $\hat{V}_{\mathcal{SE}}=\hat{V}_{A\mathcal{E}_A}+\hat{V}_{B\mathcal{E}_B}+\hat{V}_{C\mathcal{E}_C}\neq0$ has nontrivial interaction between $S$ and $E$ (i.e. contains two-body interaction terms), then $\left[P^{AB},\hat{V}_{\mathcal{SE}}\right]\neq0$.}

\begin{proof}
We can simplify $[P^{AB},\hat{V}_{\mathcal{SE}}]=[P^{AB},\hat{V}_{A\mathcal{E}_A}]+[P^{AB},\hat{V}_{B\mathcal{E}_B}]$, as $P^{AB}$ commutes with $C$.

Now, $[P^{AB},\hat{V}_{\mathcal{SE}}]$ would be zero if either:
\begin{enumerate}
    \item $[P^{AB},\hat{V}_{A\mathcal{E}_A}]\overset{!}{=}0$ and $[P^{AB},\hat{V}_{B\mathcal{E}_B}]\overset{!}{=}0$ are separately zero, or
    \item if $[P^{AB},\hat{V}_{A\mathcal{E}_A}]\overset{!}{=}-[P^{AB},\hat{V}_{B\mathcal{E}_B}]$.
\end{enumerate}
 
In general, we can write $\hat{V}_{A\mathcal{E}_A}=\sum_{ij=0,\ldots,3}g_{ij}\sigma_{i}^{A}\otimes\sigma_{j}^{\mathcal{E}_A}$, for $\sigma_{i}\in\left\{ \id,\sigma_{x},\sigma_{y},\sigma_{z}\right\} $. Then,
\begin{align}
[P^{AB},\hat{V}_{A\mathcal{E}_A}] & =-\dfrac{1}{4}\sum_{ij=0,\ldots,3}g_{ij}\left[\sigma_{x},\sigma_{i}^{A}\right]\otimes\sigma_{x}^{B}\otimes\sigma_{j}^{\mathcal{E}_A}\nonumber  -\dfrac{1}{4}\sum_{ij=0,\ldots,3}g_{ij}\left[\sigma_{y},\sigma_{i}^{A}\right]\otimes\sigma_{y}^{B}\otimes\sigma_{j}^{\mathcal{E}_A}\nonumber \\
 & \phantom{=} -\dfrac{1}{4}\sum_{ij=0,\ldots,3}g_{ij}\left[\sigma_{z},\sigma_{i}^{A}\right]\otimes\sigma_{z}^{B}\otimes\sigma_{j}^{\mathcal{E}_A}.
\end{align}
Due to the linear independence of the $\left\{ \sigma_{x}^{B},\sigma_{y}^{B},\sigma_{z}^{B}\right\} $ terms, in order for $[P^{AB},\hat{V}_{A\mathcal{E}_A}]\overset{!}{=}0$
to be true, we would require
\begin{align}
\sum_{ij=0,\ldots,3}g_{ij}\left[\sigma_{x},\sigma_{i}^{A}\right]\otimes\sigma_{x}^{B}\otimes\sigma_{j}^{\mathcal{E}_A} &\overset{!}{=}0 \label{eq:separated_sets_1}\\
\sum_{ij=0,\ldots,3}g_{ij}\left[\sigma_{y},\sigma_{i}^{A}\right]\otimes\sigma_{y}^{B}\otimes\sigma_{j}^{\mathcal{E}_A}&\overset{!}{=}0 \label{eq:separated_sets_2}\\
\sum_{ij=0,\ldots,3}g_{ij}\left[\sigma_{z},\sigma_{i}^{A}\right]\otimes\sigma_{z}^{B}\otimes\sigma_{j}^{\mathcal{E}_A}&\overset{!}{=}0.
\label{eq:separated_sets_3}
\end{align}
Consider Eq.~(\ref{eq:separated_sets_1}):
\begin{align}
\sum_{ij=0,\ldots,3}g_{ij}  \left[\sigma_{x},\sigma_{i}^{A}\right]\otimes\sigma_{x}^{B}  \otimes\sigma_{j}^{\mathcal{E}_A} =& \sum_{j=0,\ldots,3}g_{yj}2i\sigma_{z}^{A}\otimes\sigma_{x}^{B}\otimes\sigma_{j}^{\mathcal{E}_A}  +\sum_{j=0,\ldots,3}g_{zj}\left(-2i\right)\sigma_{y}^{A}\otimes\sigma_{x}^{B}\otimes\sigma_{j}^{\mathcal{E}_A}.
\end{align}
Again, due to linearly independence of the Pauli matrices, this is zero iff $g_{yj}=g_{zj}=0$ for all $j=0,\ldots,3$. The only possible nonzero terms are $g_{xj}$ and $g_{0j}$. Similarly, consider Eq. (\ref{eq:separated_sets_2}):
\begin{align}
\sum_{ij=0,\ldots,3}g_{ij} \left[\sigma_{y},\sigma_{i}^{A}\right]\otimes\sigma_{y}^{B}\otimes\sigma_{j}^{\mathcal{E}_A} = &\sum_{j=0,\ldots,3}g_{xj}\left(-2i\right)\sigma_{z}^{A}\otimes\sigma_{y}^{B}\otimes\sigma_{j}^{\mathcal{E}_A} +\sum_{j=0,\ldots,3}g_{zj}\left(2i\right)\sigma_{x}^{A}\otimes\sigma_{y}^{B}\otimes\sigma_{j}^{\mathcal{E}_A}.
\end{align}
which is zero iff $g_{xj}=g_{zj}=0$ for all $j=0,\ldots,3$. Similarly, Eq. (\ref{eq:separated_sets_3}) gives $g_{xj}=g_{yj}=0$ for all $j=0,\ldots,3$.

Hence, if $[P^{AB},\hat{V}_{A\mathcal{E}_A}]\overset{!}{=}0$ then we must have  $\hat{V}_{A\mathcal{E}_A}=\sum_{j=0,\ldots,3}g_{0j}\id_{A}\otimes\sigma_{j}^{\mathcal{E}_A}$. This is trivial because this is a local interaction on the environment $\mathcal{E}_B$ and contradicts the fact that $\hat{V}_{A\mathcal{E}_A}$ must be an interaction between system and environment. Similarly, $[P^{AB},\hat{V}_{B\mathcal{E}_B}]\overset{!}{=}0$ implies that $\hat{V}_{B\mathcal{E}_B}=\sum_{j=0,\ldots,3}\id_{B}\otimes\sigma_{j}^{\mathcal{E}_B}$ is a local interaction on the environment $\mathcal{E}_B$.

Finally, in general, 
\begin{align}
  [P^{AB},\hat{V}_{\mathcal{SE}}] & =-\dfrac{1}{4}\sum_{ij=0,\ldots,3}g_{ij}\left[\sigma_{x},\sigma_{i}^{A}\right]\otimes\sigma_{x}^{B}\otimes\sigma_{j}^{\mathcal{E}_A}\otimes\id^{\mathcal{E}_B}\nonumber \\
 & \phantom{=}-\dfrac{1}{4}\sum_{ij=0,\ldots,3}g_{ij}\left[\sigma_{y},\sigma_{i}^{A}\right]\otimes\sigma_{y}^{B}\otimes\sigma_{j}^{\mathcal{E}_A}\otimes\id^{\mathcal{E}_B}\nonumber \\
 & \phantom{=}-\dfrac{1}{4}\sum_{ij=0,\ldots,3}g_{ij}\left[\sigma_{z},\sigma_{i}^{A}\right]\otimes\sigma_{z}^{B}\otimes\sigma_{j}^{\mathcal{E}_A}\otimes\id^{\mathcal{E}_B}\nonumber \\
 & \phantom{=}-\dfrac{1}{4}\sum_{ij=0,\ldots,3}f_{ij}\sigma_{x}^{A}\otimes\left[\sigma_{x},\sigma_{i}^{B}\right]\otimes\id^{\mathcal{E}_A}\otimes\sigma_{j}^{\mathcal{E}_B}\nonumber \\
 & \phantom{=}-\dfrac{1}{4}\sum_{ij=0,\ldots,3}f_{ij}\sigma_{y}^{A}\otimes\left[\sigma_{y},\sigma_{i}^{B}\right]\otimes\id^{\mathcal{E}_A}\otimes\sigma_{j}^{\mathcal{E}_B}\nonumber \\
 & \phantom{=}-\dfrac{1}{4}\sum_{ij=0,\ldots,3}f_{ij}\sigma_{z}^{A}\otimes\left[\sigma_{z},\sigma_{i}^{B}\right]\otimes\id^{\mathcal{E}_A}\otimes\sigma_{j}^{\mathcal{E}_B}.
\end{align}
All the terms where $j=x,y,z$ are separately linearly independent (i.e. each line above is linearly independent from the others). Analously, we find that this is zero iff $g_{xj}=g_{yj}=g_{zj}=0$ for $j=x,y,z$ and $f_{xj}=f_{yj}=f_{zj}=0$ for $j=x,y,z$. This leaves us with $\hat{V}_{A\mathcal{E}_A}=\sum_{j=0,\ldots,3}g_{0j}\id_{A}\otimes\sigma_{j}^{\mathcal{E}_A}+\sum_{i=0,\ldots,3}g_{i0}\sigma_{i}^{A}\otimes\id^{\mathcal{E}_A}$ and similarly for $\hat{V}_{B\mathcal{E}_B}$, i.e. they contain only local terms, which is a contradiction.
\end{proof}

\subsection{Proof of Lemma~\ref{lem:H_B_H_0_V}\label{sec:proof_of_lem_HB_H0_V}}

Recall that we want to show that $[\hat{H}_{\vec{\boldsymbol{B}}} (\theta,\phi),\hat{H}_0+\hat{V}_{\mathcal{SE}}]\neq 0$ for all nontrivial $\hat{V}_{\mathcal{SE}}$.

\begin{proof}
Now, 
\begin{equation}
\left[\hat{H}_0+\hat{V}_{\mathcal{SE}},\hat{H}_{\vec{\boldsymbol{B}}}\left[\theta,\phi\right]\right]=\left[\hat{H}_0,\hat{H}_{\vec{\boldsymbol{B}}}\left(\theta,\phi\right)\right]+\left[\hat{V}_{\mathcal{SE}},\hat{H}_{\vec{\boldsymbol{B}}}\left(\theta,\phi\right)\right].
\end{equation}
The first term on the right hand side contains only nonzero terms on the system $ABC$ while the $\hat{V}_{\mathcal{SE}}$ has environment terms. For this to be zero overall, it must be that any nontrivial environment terms in the second commutator are zero. 
\begin{equation}
\left[\hat{V}_{\mathcal{SE}},\hat{H}_{\vec{\boldsymbol{B}}}\left(\theta,\phi\right)\right]=\left[\hat{V}_{A\mathcal{E}_A},\hat{H}_{\vec{\boldsymbol{B}}|A}\right]+\left[\hat{V}_{B\mathcal{E}_B},\hat{H}_{\vec{\boldsymbol{B}}|B}\right]+\left[\hat{V}_{C\mathcal{E}_C},\hat{H}_{\vec{\boldsymbol{B}}|C}\right].
\end{equation}
Clearly each term is independent, so let's just look at the first (and the other two follow analogously):
\begin{align}
\left[\hat{V}_{A\mathcal{E}_A},\hat{H}_{\vec{\boldsymbol{B}}|A}\right] & =\left[\sum_{ij=0,\ldots,3}g_{ij}\sigma_{i}^{A}\otimes\sigma_{j}^{\mathcal{E}_A},\hat{H}_{\vec{\boldsymbol{B}}|A}\right] =\sum_{ij=0,\ldots,3}g_{ij}\left[\sigma_{i}^{A},\hat{H}_{\vec{\boldsymbol{B}}|A}\right]\otimes\sigma_{j}^{\mathcal{E}_A}.
\end{align}
For the indices where $j\neq0$ (i.e. for the non-identity terms on the environment $\mathcal{E}_A$), we want the coefficients to be zero, i.e. 
\begin{equation}
\sum_{i=0,\ldots,3}g_{ij}\left[\sigma_{i}^{A},\hat{H}_{\vec{\boldsymbol{B}}|A}\right]\otimes\sigma_{j}^{\mathcal{E}_A}\overset{!}{=}0\quad\forall j=x,y,z.
\end{equation}
If $\hat{H}_{\vec{\boldsymbol{B}}|A}=\dfrac{\gamma B_{0}}{2}\sigma_{x}$ for example, then this would imply that 
\begin{align}
0 & \overset{!}{=}\sum_{i=0,\ldots,3}g_{ij}\gamma B_{0}\left[\sigma_{i}^{A},\sigma_{x}\right]  =g_{yj}\gamma B_{0}\left(-2i\sigma_{z}\right)+g_{zj}\gamma B_{0}2i\sigma_{y}\\
\implies & g_{yj}=g_{zj}=0\quad\forall j=x,y,z.
\end{align}
We can similarly show that $g_{xj}=0$ for all $j=x,y,z$, in which case $V_{A\mathcal{E}_A}=\id_{A}\otimes\sum_{j=0,x,y,z}g_{0j}\sigma_{j}^{\mathcal{E}_A}+\sum_{i=x,y,z}\sigma_{i}^{A}\otimes\id^{\mathcal{E}_A}$, i.e. containing only local terms.
Similarly for $\hat{V}_{B\mathcal{E}_B}$ and $\hat{V}_{C\mathcal{E}_C}$. This is a contradiction because $\hat{V}_{\mathcal{SE}}$ must be a nontrivial interaction. Hence, provided that $\hat{V}_{\mathcal{SE}}$ is a nontrivial interaction, we always have $[\hat{H}_0+\hat{V}_{\mathcal{SE}},\hat{H}_{\vec{\boldsymbol{B}}}\left[\theta,\phi\right]]\neq0$.
\end{proof}
\end{widetext}

\section{Another system-environment interaction: the SWAP\label{app:SWAP}}

What of more general system-environment interactions and more general environment initial states? If so, the condition of $\rho_{\mathcal{S}}\left(0\right)\neq\id_{\mathcal{S}}/d_{\mathcal{S}}$ is not necessary. In the \hyperref[sec:discussion]{Discussion} in the main text, we mentioned that the ultimate generation of coherence is key. In the following example, we can do this by transferring environment coherence into the system.

\begin{figure*}
\begin{centering}
\includegraphics[width=1.0\textwidth]{"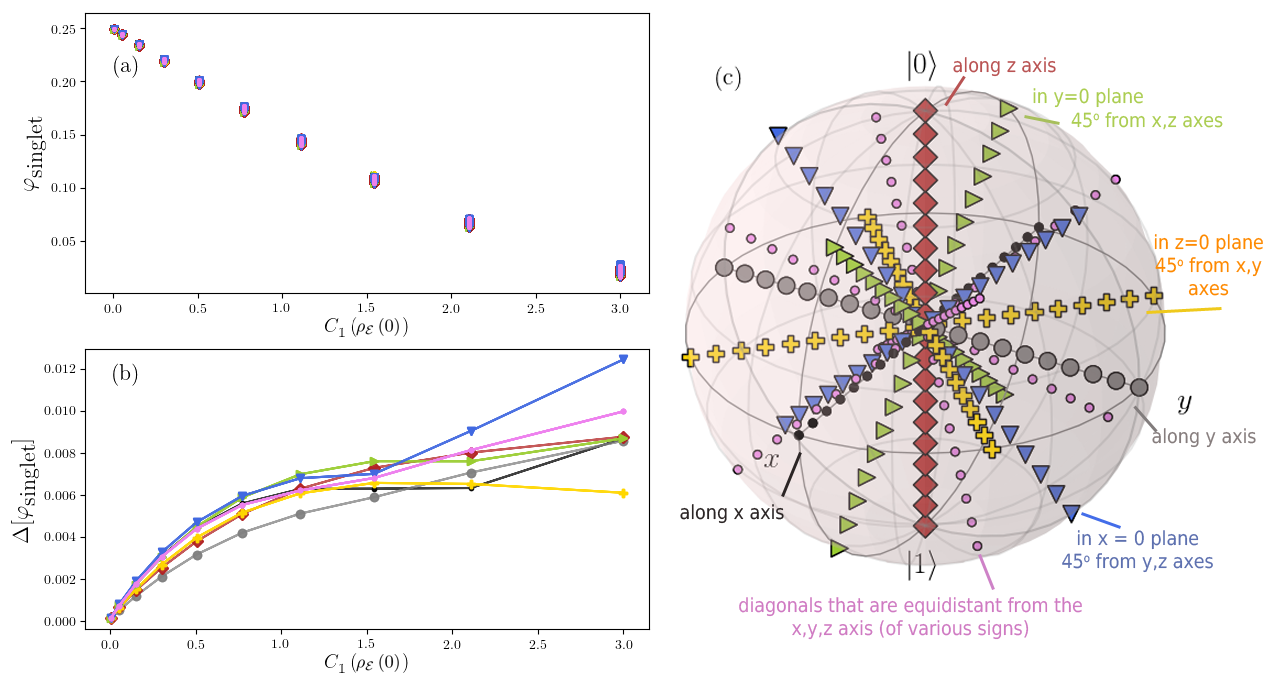"}
\par\end{centering}
\caption{Magnetic sensor performance for a \textbf{SWAP interaction} where the system is initially maximally mixed $\rho_{\mathcal{S}}\left(0\right)=\id/d_{\mathcal{S}}$ but the environment is not (see Example \ref{example:SWAP}). (a) The recombination value $\varphi_{\text{singlet}}$ for 500 randomly chosen initial $\rho_{\mathcal{E}}\left(0\right)$ versus the initial basis-independent coherence $C_{\id}\left(\rho_{\mathcal{E}}\left(0\right)\right)$ (b) Absolute anisotropy $\Delta\left[\varphi_{\text{singlet}}\right]$. (c) Visualisation of the initial environment states on the Bloch sphere. Some points have the same colouring and shape, as their behaviour in the absolute anisotropy follow the same trend. Model and all other parameters not mentioned explicitly here listed in Table \ref{tab:Parameters-and-configurations}.\label{fig:l1coherence_vs_varphi_environment}}
\end{figure*}

\begin{example}
\label{example:SWAP}The initial state $\rho_{\mathcal{S}}\left(0\right)=\id/d_{\mathcal{S}}$ on the system can lead to nontrivial sensing if (for example) $\hat{V}_{\mathcal{SE}}$ is a SWAP interaction between system and environment, and the initial environment states are not maximally mixed.  The SWAP interaction will introduce the environment coherence into the system. Consider the isotropic interaction
\begin{equation}
\hat{V}_{X\mathcal{E}_{X}}=J_{se} \sum_{i=x,y,z}\sigma_{i}^{X}\otimes\sigma_{i}^{\mathcal{E}_{X}}, \quad \text{for }X=A,B,C.
\end{equation}
Then $\exp[-it_{SE}\hat{V}_{A\mathcal{E}_A}]$ performs a perfect SWAP if $J_{se}t_{se}=\pi/4$ \citep{Pezzutto2019}. We can see in Fig. \ref{fig:l1coherence_vs_varphi_environment} that the coherence in the environment leads to a small change in the recombination value and hence a nontrivial magnetic sensor.
\end{example}

Fig. \ref{fig:l1coherence_vs_varphi_environment}(b) also shows that there is an initial sharp rise in the sensor perform versus initial environment coherence, before a partial plateau emerges. Maximum initial coherence does not lead to the same performance--there is some dependence on the basis. However, these results suggest that small environment coherences are most effective, in which case the basis does not matter as much.

Note that non-maximally-mixed states means that the effective interaction between system and environment will not reduce to the standard decohering Lindblad master equation.

This example shows that \emph{environment} coherence can contribute to the performance of the magnetic sensor. Here, the environment coherence introduces effective non-unital dynamics on the system, which then creates anisotropy in the system state\textemdash which then allows the magnetic field to affect it.


\section{Coherence yield \label{sec:coherence_yield_results}}

In the main text, we compared how the \emph{initial} coherence affected the subsequent magnetic sensor performance. However, it is also interesting to examine the coherence \emph{yield}, i.e. how much coherence is relevant during the recombination process itself.

Due to the trace-decreasing nature of $\rho_{\mathcal{SE}}\left(t\right)$ over time in our model, we will use the modified version as the correct distance to the trace-decreasing maximally mixed state:
\begin{align}
C_{\id}^{*}\left(\rho\right) & =S\left(\rho\biggl|\biggr|\dfrac{\id_{d}}{d}\tr\rho\right)\\
 & =\left(\log_{2}d-\log_{2}\tr\rho\right)\tr\rho-S\left(\rho\right).
\end{align}

\begin{figure*}
\begin{centering}
\includegraphics[width=1\textwidth]{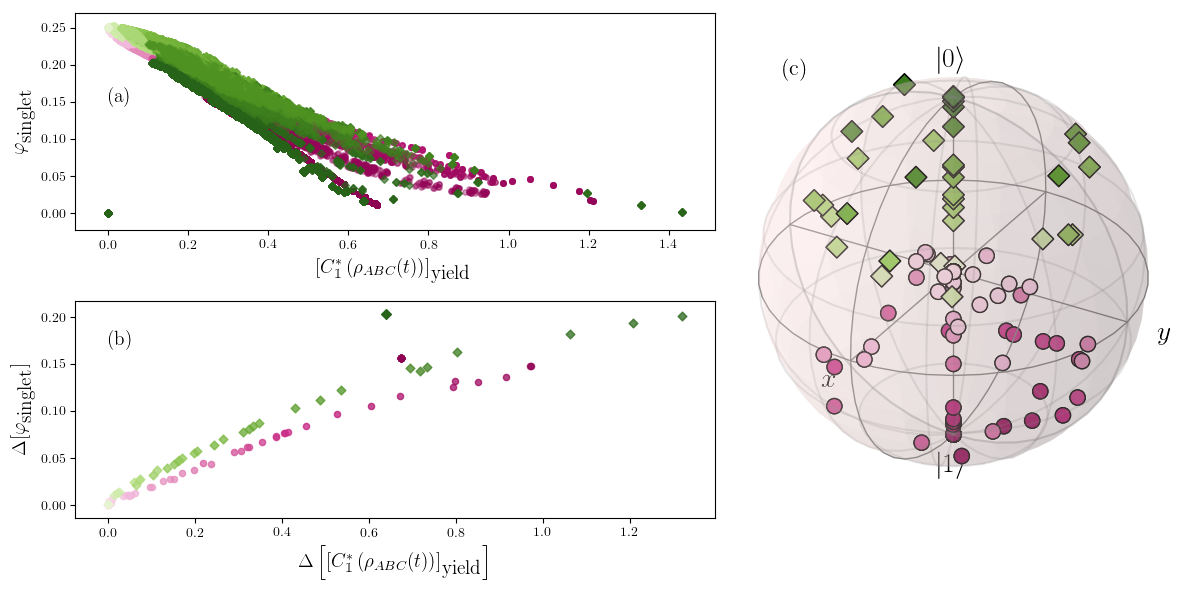}
\par\end{centering}
\caption{Relationship between the basis-independent quantum
coherence yield vs sensor performance. (a) Yield of basis-independent coherence $\left[C_{\id}^{*}\left(\rho_{ABC}\right)\right]_{\text{yield}}$ vs singlet recombination yield $\varphi_{\text{singlet }}$. (b) Absolute anisotropy of the basis-independent coherence $\Delta\left[C_{\id}^{*}\left(\rho_{ABC}\right)\right]_{\text{yield}}$ vs absolute anisotropy of the singlet recombination yield $\Delta\left[\varphi_{\text{singlet }}\right]$.
(c) Bloch sphere depicting the initial state $\rho_A$, where $\rho_{ABC}(0) = \rho_A \otimes \rho_A \otimes \id/2$.
All other model parameters listed in Table \ref{tab:Parameters-and-configurations}. \label{fig:basis-independent-coherence}}
\end{figure*}

Fig. \ref{fig:basis-independent-coherence} shows the basis-independent coherence yield versus the singlet recombination yield. There is a clear association between the two: greater coherence yield leads to a \emph{lower} recombination. However, the range of these yields is important for the sensor performance. We see also in  Fig. \ref{fig:basis-independent-coherence} that in which there is a strong positive linear relationship between the anisotropies of coherence and recombination yields. The subtrends can be well separated by considering the sign of the $z$ coordinate on the Bloch sphere, i.e. $\tr[\sigma_z \rho_A (0)]$, of the initial state.

These results show that basis-independent coherence, both initial and during, suppresses the raw value of the recombination yield. However, larger initial coherence also leads to greater \emph{variation} in both recombination yield and coherence yield.


\section{System-environment correlations and quantum Darwinism \label{sec:quantum_Darwinism}}

One of the problems with quantum coherence is that it is typically basis-dependent. One way to get around the basis-dependence is to consider basis-independent coherence, which we did in the main paper. A second solution is to consider correlations between states. In this appendix, we analyse some of the quantum discord and classical correlations between the system radicals and the environment.

The quantum mutual information between two systems $X,Y$ with joint state $\rho_{XY}$ is 
\begin{align}
I\left(X:Y\right) & =S\left(\rho_{X}\right)+S\left(\rho_{Y}\right)-S\left(\rho_{XY}\right),
\end{align}
where the von Neumann entropy is $S(\rho) = -\tr[\rho \log\rho]$, and the reduced states are $\rho_{X}=\tr_{Y}\left[\rho_{XY}\right]$ and $\rho_{Y}=\tr_{X}\left[\rho_{XY}\right]$. The quantum mutual information describes the correlations between $X$ and $Y$, and can be decomposed into quantum discord and classical information. The informational measure of quantum discord \citep{Henderson2001,Ollivier2001} is 
\begin{equation}
\mathcal{D}\left(X:Y\right)=\min_{\left\{ \Pi_{X|i}\right\} _{i}}\left[\sum_{i}p_{i}S\left(\rho_{X|i}\right)+S\left(\rho_{X}\right)-S\left(\rho_{XY}\right)\right],
\end{equation}
where $\rho_{Y|i}$ is the conditional state on $Y$ after measurement result $i$ on $X$ using the von Neumann measurement $\left\{ \Pi_{X|i}\right\} _{i}$ with probability $p_{i}$:\footnote{The measurement can be generalised to positive-operator valued measures (POVMs).}
\begin{align}
\rho_{B|i} & =\dfrac{\tr_{X}\left[\Pi_{X|i}\rho_{XY}\left(\Pi_{X|i}\right)^{\dagger}\right]}{p_{i}},\\
p_{i} & =\tr\left[\Pi_{X|i}\rho_{XY}\left(\Pi_{X|i}\right)^{\dagger}\right].
\end{align}
The classical Holevo information can then be defined as the difference between the quantum mutual information and the quantum discord \citep{Zwolak2013}:
\begin{align}
\chi\left(X:Y\right) & =I\left(X:Y\right)-\mathcal{D}\left(X:Y\right).
\end{align}

In quantum Darwinism \citep{Zurek2009}, there is an expanded notion of classicality called \emph{objectivity}: a system is objective when its information is broadcasted to the environment \citep{Horodecki2015}. Objectivity (in the framework of strong quantum Darwinism \citep{Le2019}) is characterised by zero quantum discord and maximal classical correlations between the system and environment\citep{Le2019}:
\begin{equation}
\mathcal{D}\left(\mathcal{S}:\mathcal{E}\right)=0,\quad\chi\left(\mathcal{S}:\mathcal{E}\right)=S\left(\rho_{\mathcal{S}}\right).
\end{equation}
We say that a system is \emph{non-objective} if it has nonzero quantum discord with the environment or non-maximal classical correlations with the environment.

\begin{figure*}
\begin{centering}
\includegraphics[width=0.85\textwidth]{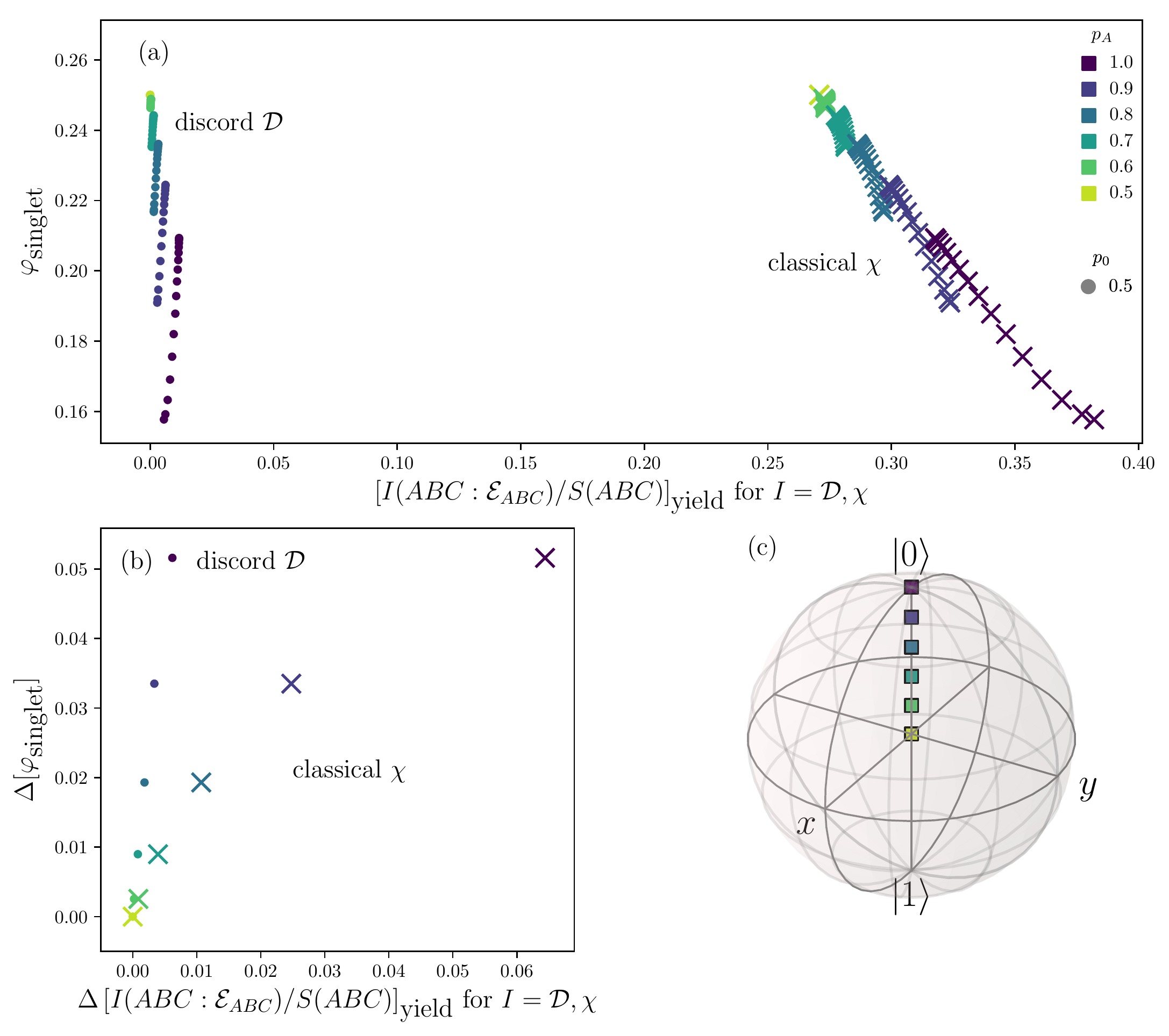}
\par\end{centering}
\caption{Relationship between system-environment correlations and sensor performance. (a) Yield of normalised classical and quantum information yields between radicals $ABC$
and their environments $\mathcal{E}_{ABC}= \mathcal{E}_A \mathcal{E}_B \mathcal{E}_C$, $\left[I\left(ABC:\mathcal{E}_{ABC}\right)/S\left(ABC\right)\right]_{\text{yield}}$, versus the singlet recombination yield $\varphi_{\text{singlet}}$, where $\chi$ denotes the classical Holevo information and $\mathcal{D}$ denotes the quantum discord. Note that the discord yield very close to zero.
(b) Absolute anisotropy of the classical and quantum information yields, $\Delta\left[I\left(ABC:\mathcal{E}_{ABC}\right)/S\left(ABC\right)\right]_{\text{yield}}$, versus recombination yield $\Delta\left[\varphi_{\text{singlet }}\right]$.
(c) Depiction of the initial states on the Bloch sphere, i.e. $\rho_{ABC} (0) = \rho_A \otimes \rho_A \otimes \id/2$, where $\rho_A = p_A |0\rangle\langle 0| + (1-p_A)|1\rangle\langle 1|$. The initial environment states are $\rho_{n,E}=(\id/2)^{\otimes 3}$. All other model parameters are listed in Table \ref{tab:Parameters-and-configurations}.  \label{fig:Information-yields-vs-varphi} }
\end{figure*}

In Fig. \ref{fig:Information-yields-vs-varphi}(a) and Fig. \ref{fig:Information-yields-vs-varphi_dif_E}(a) we plot the normalised classical and quantum information yields with the singlet recombination, where the information yield of $I=\chi,\mathcal{D}$ is defined analogously to Eq.~(\ref{eq:coherence_yield}):
\begin{equation}
I_{\text{yield}}\left[\theta,\phi\right]=\dfrac{k}{\varphi_{\text{singlet}}}\int_{0}^{\infty}d\tau  I\left(\tau\right)\tr[P_{\text{singlet}}^{\left(A,B\right)}\rho_{\mathcal{S}}\left(\tau\right)].
\end{equation}
The initial system-environment states are product with each other $\rho_{\mathcal{S}}\otimes\rho_{\mathcal{E}}$, hence have no initial correlations.
We use the information yield: we can measure the amount of correlations generated during evolution ($\chi(t),\mathcal{D}(t)$), modulated by how much it contributed to the final recombination yield (which is done by the $tr[P^{(A,B)}_\text{singlet}\rho_S(t)]$ term). 

We consider the information between system radicals $ABC$ and the environment spins $\mathcal{E}_{ABC}$. We compare two different initial environment states, $\rho_\mathcal{E} = \id/2$ like in the rest of the paper in Fig. \ref{fig:Information-yields-vs-varphi}, versus $\rho_\mathcal{E}=\ket{0}\bra{0}$ in Fig. \ref{fig:Information-yields-vs-varphi_dif_E}. This is because different environment states are known to affect subsequent quantum Darwinism \cite{Giorgi2015,Zwolak2009,Zwolak2010,Balaneskovic2015,Balaneskovic2016}.

In Fig. \ref{fig:Information-yields-vs-varphi}(a), the quantum discord yield is very close to zero, while the classical information is \emph{not} maximal (maximum normalised classical information will have value $=1$). Thus, the radicals $ABC$ are non-objective and close to ``classical'' in the sense of lacking quantum correlations with their environment. Furthermore, the greatest classical information yield  corresponds to the lowest recombination yield in the bottom right of Fig.~\ref{fig:Information-yields-vs-varphi}(a). This suggests that a fully objective state would be useless as a quantum magnetic sensor: if the state is fully objective, then the environment is perfectly monitoring the system, corresponding to a quantum-Zeno-like situation. Thus, non-objectivity appears to be necessary for a nontrivial magnetic sensor.

In Fig. \ref{fig:Information-yields-vs-varphi}(b) we plot the absolute anisotropy for the classical and quantum correlation yields versus the absolute anisotropy of the singlet recombination. It shows that greater variation in classical information is correlated with greater sensor performance. There is a relationship between (non)-objectivity and sensor performance, which can be explained as follows: the initial states with the greatest basis-independent coherence vary more across different field angles, leading to greater variation in the recombination yield and greater variation in the classical correlations. 

This leads to an interesting point: the spins in the magnetic sensor can happily operate in the space in which it is instantaneously classical (i.e. lacking/close-to-zero quantum correlations), and yet non-objective. This is tempered by the results of the main text, where the spins contain quantum coherence, both in the initial state and during evolution. Thus, these results pin-point the level of ``quantumness'' in the magnetic sensor.

\begin{figure*}
\begin{centering}
\includegraphics[width=0.85\textwidth]{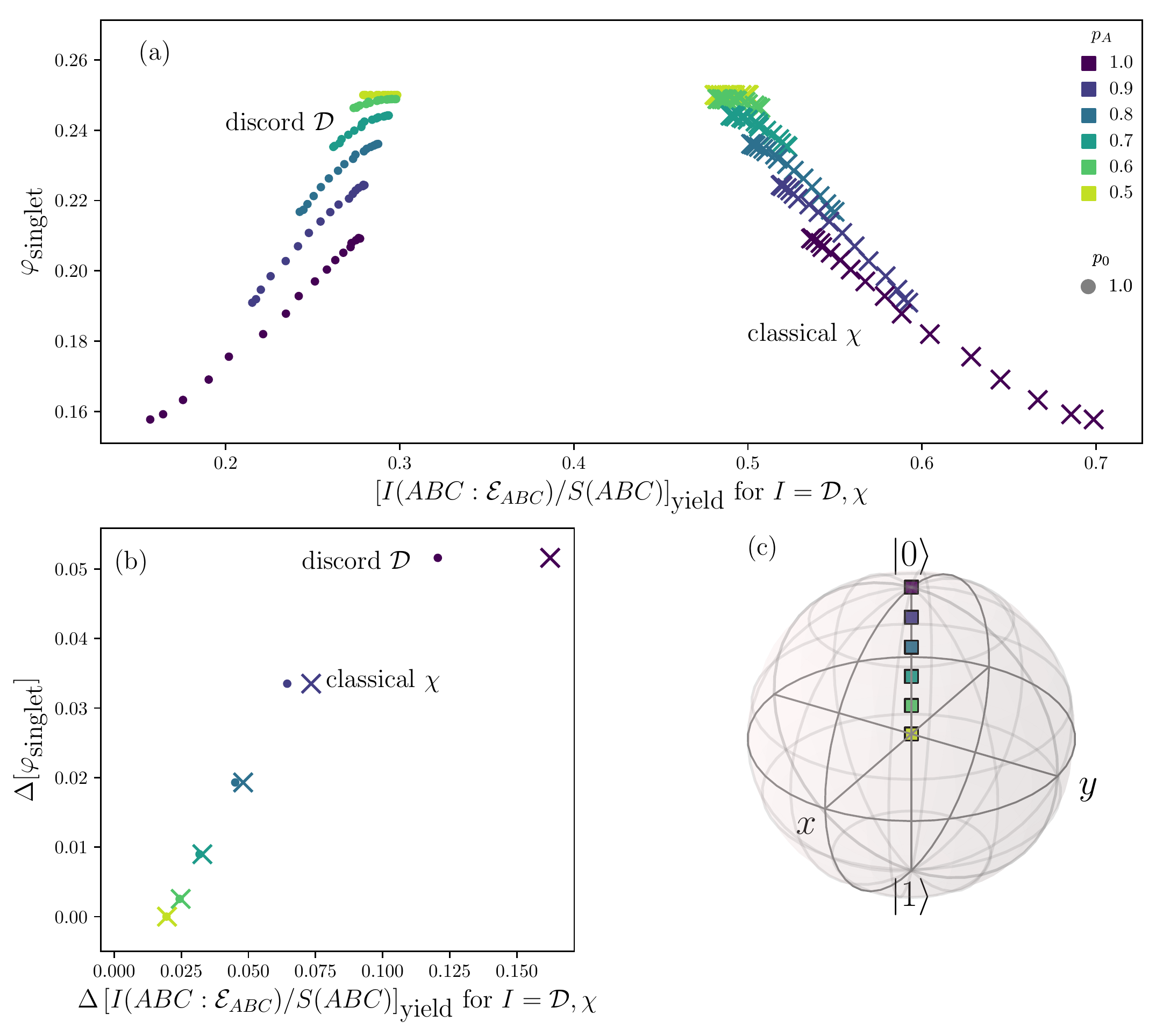}
\par\end{centering}
\caption{Relationship between system-environment correlations and sensor performance, with pure environment states. Exactly as Fig.~\ref{fig:Information-yields-vs-varphi}, except that the initial environment states are $\rho_{n,E}=|0\rangle\langle 0|^{\otimes 3}$.
(a) Yields of normalised information versus singlet recombination yield. (b) Absolute anisotropy of information versus absolute anisotropy of recombination yield. (c) Depiction of system initial states  $\rho_{ABC} (0) = \rho_A \otimes \rho_A \otimes \id/2$, where $\rho_A = p_A |0\rangle\langle 0| + (1-p_A)|1\rangle\langle 1|$ on the Bloch sphere.\label{fig:Information-yields-vs-varphi_dif_E} }
\end{figure*}

In contrast, in Fig.~\ref{fig:Information-yields-vs-varphi_dif_E} where the initial environment states are pure, $\ket{0}\bra{0}$, there is much more quantum discord over time. However, if we compare Fig.~\ref{fig:Information-yields-vs-varphi}(b) Fig.~\ref{fig:Information-yields-vs-varphi_dif_E}(b), the sensor performance given by $\Delta[\varphi_{\text{singlet}}]$ are very similar. This shows that in the model considered, quantum magnetic sensor performance relies primarily on the initial coherence in the system radicals. This is primarily due to how the magnetic field does not directly interact with the environment state. (Not shown here: when the magnetic field acts on the environment, the environment initial coherence affects the recombination yield.)

The combined results of Fig.~\ref{fig:Information-yields-vs-varphi} and  Fig.~\ref{fig:Information-yields-vs-varphi_dif_E} show that the quantum magnetic sensor always operates in a non-objective state with its environment. Hence, the presence of non-objectivity alone is not sufficient for sensor performance, though the numerical results suggest that non-objectivity may be necessary. While there is a relationship between non-objectivity and quantum sensor performance, sensor performance is much more strongly linked to initial system coherence in our model. 

In summary, the results of this appendix show the magnetic sensor can operate in a regime that has quantum coherence and is non-objective, but performance is better predicted by initial system coherence rather than system-environment correlations in our model.

\end{document}